\documentclass[11pt]{amsart}
\usepackage{amssymb,latexsym, amsmath, amsxtra}
\usepackage[dvips]{graphics}
\usepackage{epsfig}

\usepackage[usenames, dvipsnames]{color} 
\usepackage[table]{xcolor}  
\usepackage{amssymb}
\usepackage{amsmath}
\usepackage{tikz}
\usepackage{amsthm}
\usepackage{mathtools}
\usepackage{braket}
\usepackage{physics}
\usepackage{bbm}
\usepackage{stmaryrd}
\usepackage{float}
\usepackage{graphicx}
\usepackage{times}
\usepackage{mdframed}
\usepackage{caption}
\usepackage{subcaption}
\usetikzlibrary{backgrounds}
\usepackage{empheq}
\usepackage{MnSymbol,wasysym}
\usepackage{mleftright}
\usepackage{cleveref}
\usetikzlibrary{arrows,matrix,positioning}

\usetikzlibrary{calc,shapes}

\newtheorem{conjecture}{Conjecture}
\newtheorem{theorem}{Theorem}[section]

\newtheorem{lemma}[theorem]{Lemma}
\newtheorem{proposition}[theorem]{Proposition}

\newcommand{\widesim}[2][1.5]{
  \mathrel{\underset{#2}{\scalebox{#1}[1]{$\sim$}}}
}

\begin{document}
\title{Mixed moments of characteristic polynomials of random unitary matrices}
\abstract Following the work of Conrey, Rubinstein and Snaith \cite{kn:crs06} and Forrester and Witte \cite{kn:forwit06} we examine a mixed moment of the characteristic 
polynomial and its derivative for matrices from the unitary group $U(N)$ (also known as the CUE) and relate the moment to the solution of a Painlev{\'e} differential equation.  We also calculate a simple form for the asymptotic behaviour of moments of logarithmic derivatives of these characteristic polynomials evaluated near the unit circle. 
\endabstract

\thanks{The authors thank the American Institute of Mathematics for hosting us at the workshop {\it Painlev{\'e} Equations and Their Applications} where this work was started. ECB is grateful to the Heilbronn Institute for Mathematical Research for support. AP is supported by Russian Science Foundation grant No.17-11-01126.  JBC is supported in part by a grant from the NSF. SB is partially supported by PRIN 2015 ``Number Theory and Arithmetic Geometry"}

\author{E. C. Bailey} \address{School of Mathematics,
University of Bristol, Bristol, BS8 1TW, United Kingdom}
\email{emma.bailey@bristol.ac.uk}

\author{S. Bettin} \address{DIMA - Dipartimento di Matematica
Via Dodecaneso, 35
16146 Genova, Italy} \email{bettin@dima.unige.it}

\author{G. Blower} \address{ Department of Mathematics and Statistics,
Fylde College,
Lancaster University,
Lancaster, LA1 4YF,
United Kingdom} \email{g.blower@lancaster.ac.uk}

\author{J.B. Conrey} \address{American Institute of Mathematics, 600 East Brokaw Road
San Jose, CA 95112, USA and 
School of Mathematics,
University of Bristol, Bristol, BS8 1TW, United Kingdom
} \email{conrey@aimath.org}

\author{A. Prokhorov} \address{Deprtment of Mathematical Sciences, Indiana University-Purdue University Indianapolis, 402 N. Blackford, LD 270
Indianapolis, IN 46202, USA \&  Saint-Petersburg State University, Universitetskaya emb. 7/9, 199034, St. Petersburg, Russia. } \email{aprokhor@iupui.edu}

\author{M.O. Rubinstein} \address{Pure Mathematics
University of Waterloo
200 University Avenue West
Waterloo, Ontario, Canada N2L 3G1} \email{mrubinstein@uwaterloo.ca}

\author{N.C. Snaith}
\address{School of Mathematics,
University of Bristol, Bristol, BS8 1TW, United Kingdom}
\email{N.C.Snaith@bris.ac.uk}

\maketitle

\section{Introduction}
Moments of characteristic polynomials  and their derivatives have been investigated in several recent papers on random matrix theory. Part of the interest in these calculations is the similarity of these moments with the corresponding averages of number theoretical functions: the Riemann zeta function and other $L$-functions \cite{kn:abprw,kn:c,kn:basor_et_al18,kn:crs06,kn:consna08,kn:dehaye10,kn:dffhmp,kn:hughes00,kn:massna16,kn:mezzadri03,kn:winn12}.   It was shown by Forrester and Witte \cite{kn:forwit06} that the leading order coefficient for moments of derivatives of unitary characteristic polynomials, derived by Conrey, Rubinstein and Snaith  \cite{kn:crs06}  is related to the solution of a version of the Painlev{\'e} III$'$ differential equation.  Using the techniques of \cite{kn:crs06} we determined a similar relation, Theorem \ref{thm:1}, for mixed moments featuring both the characteristic polynomial and its derivative.  This theorem is proved in Section \ref{sect:theorem1}.  Subsequently this result also featured in work by the group \cite{kn:basor_et_al18} (see their equation 5-79), although they use different methods, allowing them to extend the result to finite $N$.    In Section \ref{sect:theorem2} we turn to the moments of the logarithmic derivative of the characteristic polynomial.  An exact formula for these moments averaged over $U(N)$ is presented in \cite{kn:consna08}, but the asymptotics when $N$ is large and the characteristic polynomials are evaluated close to the unit circle are not easy to extract from that result, whereas adapting the method of Section \ref{sect:theorem1} allows us to work out the leading order term. In Section \ref{sect:comparison} we compare this with the exact result in a couple of simple cases. In Section \ref{sect:rhp} we put the moment determinant appearing during the asymptotic analysis in the framework of Riemann-Hilbert problems.

Our definition of the  characteristic polynomial
$\Lambda_X(s)$ of $X \in U(N)$ is 
\begin{equation}\label{1.1}
\Lambda_X(s) = \det ( I - s X^*) = \prod_{j=1}^N (1 - s e^{-i \theta_j}),
\end{equation}
with the eigenvalues of $X$ denoted by
$e^{i \theta_1}, \dots ,e^{i \theta_N},$ and $X^*$ being the conjugate transpose.  

A related quantity is 
\begin{equation}\label{eq:Z}
Z_X(s) =e^{-\pi i N/2} e^{i\sum_{n=1}^N\theta_n/2} s^{-N/2} \Lambda_X(s).
\end{equation}
This definition makes $Z_X(s)$ real on the unit circle. 
Also, $Z_X(s)$ satisfies the following functional equation,
\begin{equation}
Z_X(s)=(-1)^N Z_{X^*} (1/s).
\end{equation}

We first summarise some related work on mixed moments involving both the characteristic polynomial and its derivative.  We start with a result from the thesis of Chris Hughes \cite{kn:hug01}.  We note that he uses slightly different notation to us:  his $Z_X(\theta)$ is the same as our $\overline{\Lambda_X(e^{i\theta})}$ and his $V_X(\theta)$ is our $\overline{Z_X(e^{i\theta})}$. For consistency, we will translate his results into our notation.  Hughes considers the quantity
\begin{eqnarray}
\tilde{F}_N(h,k)&:=&\int_{U(N)} |V_X(0)|^{2k-2h} |V_X'(0)|^{2h} dX_N \nonumber\\
&=& \int_{U(N)} |Z_X(1)|^{2k-2h}|Z'_X(1)|^{2h} dX_N,
\end{eqnarray}
where the average is over Haar measure on the unitary group and in the final line (and for the remainder of the paper) we are using the definition of $Z_X$ given at (\ref{eq:Z}).

Hughes shows that
\begin{equation}
\lim _{N\rightarrow \infty} \frac{1}{N^{k^2+2h}} \tilde{F}_N(h,k) = \tilde{F}(h,k),
\end{equation}
where $\tilde{F}(h,k)$ is given as an expression that is analytic in $k$ for $\Re k>h-1/2$, but the method forces $h$ to be an integer. By computing some specific examples, Hughes suggests that for a given integer $h$, $\tilde{F}(h,k)$
has the form of a rational function of $k$ multiplied by a ratio of Barnes G-functions.   Dehaye \cite{kn:dehaye10} proved this form for $\tilde{F}(h,k)$, and gave further information about the structure of the rational function of $k$, but still always for integer $h$. Winn \cite{kn:winn12} has given the only example we know of where the exponent on the derivative is not an even integer, by writing down an explicit formula $\tilde{F}_N(h,k)$ when $h=(2m-1)/2$ for $m\in \mathbb{N}$.

The asymptotics of a similar mixed moment, with just a first power on the derivative of the characteristic polynomial, but for non-integer powers on the characteristic polynomial itself, has been studied in the thesis of Ian Cooper \cite{kn:coo18} when the average is over the classical compact groups $SO(2N)$ and $USp(2N)$.

Note, there is interest in allowing the power on the derivatives of the characteristic polynomial to be non-integer, but this appears to be a difficult problem.

In this paper we will prove the following theorems. The first theorem expresses the mixed moments of $Z$ and $Z'$
in terms of derivatives of a determinant involving the $I$-Bessel function:
\begin{theorem}
\label{thm:1}
For $K, M$ integers with $2K\geq 2M\geq 0$, we have
\begin{eqnarray}
\notag
&&\int_{U(N)}|Z_X^\prime(1)|^{2K-2M}|Z_X(1)|^{2M}dX_N=(-1)^{K(K-1)/2+K-M} \\
&&\qquad \times N^{K^2+2K-2M} \left( \frac{d}{dx}\right) ^{2K-2M} \left( e^{-x/2}x^{-K^2/2} \det_{K\times K}\left( I_{i+j-1}(2\sqrt{x})\right)\right)\big(1+O\big(\tfrac{1}{N}\big)\big) \Bigg|_{x=0}.
\label{eq:generalization of CRS}
\end{eqnarray}
This can further be written in terms of a solution to a Painlev\'e equation, as expressed in~\eqref{eq:painleve}.
\end{theorem}

Our second theorem gives the leading asymptotics of the moments of the logarithmic derivative of $\Lambda$
at a point approaching the unit circle:
\begin{theorem}
\label{thm:2}
Let $\Re \alpha>0$ and $K\in\mathbb{N}$,
\begin{equation}
   \int_{U(N)}\left|\frac{\Lambda_X^\prime}{\Lambda_X}(e^{-\alpha})\right|^{2K}dX_N
   =
    \binom{2K-2}{K-1}
    \frac{N^{2K}}{(2a)^{2K-1}}(1+O(a)),
    \label{eq:final asymptotic}
\end{equation}
where $\alpha=a/N$ and $a=o(1)$ as $N\rightarrow \infty$ (so that $\alpha$ depends on $N$).
\end{theorem}

These two theorems lead us immediately to conjectures about mixed moments for the Riemann zeta-function.
Recall that the Riemann zeta-function 
\begin{equation}\zeta(s)=\sum_{n=1}^\infty \frac{1}{n^s}  \qquad (s=\sigma+it, \sigma>1)\end{equation}
satisfies the functional equation
\begin{equation}\xi(s)=\frac 12 s(s-1) \pi^{-\frac s 2}\Gamma\left(\frac s 2\right)\zeta(s)=\xi(1-s)  \end{equation}
 where $\xi(s)$ is entire. Therefore Hardy's function 
 \begin{equation}Z(t)=\frac{ \pi^{-\frac {it}2}\Gamma\left(\frac 14 +\frac {it} 2 \right)}{|  \Gamma\left(\frac 14 +\frac {it} 2 \right)|} \zeta(1/2+it)\end{equation}
 is real for real $t$ and satisfies $|Z(t)|=|\zeta(1/2+it)|$.  Our Theorem \ref{thm:1} involves $Z_X$ which is the random matrix   analogue of Hardy's $Z(t)$ function and Theorem \ref{thm:2} involves $\Lambda_X$ which is the 
 random matrix analogue of $\zeta$. The conjecture  of Keating and Snaith  \cite{kn:keasna00a}
 about moments of the Rieman zeta-function may be written as 
 \begin{equation}\frac 1 T \int_0^T |Z(t)|^{2K}~dt \sim \prod_{j=0}^{K-1}\frac{j!}{(j+K)!}a_K (\log T)^{K^2}\end{equation}
 for a certain arithmetic constant $a_K$. After the work of Hughes \cite{kn:hug01} and Conrey, Rubinstein and Snaith \cite{kn:crs06}
 we expect that the $2K$th moment of $|Z'(t)|$ involves the same arithmetic constant $a_K$  multiplied 
 by a (rational number) geometric factor and $(\log T)^{K^2+2K}$.  These ideas translate to a conjecture  for the 
 mixed moments we are considering here. We can express our conjecture as follows.
 \begin{conjecture} For non-negative integers $K$ and $M$ with $M\le K$ we conjecture that  as $T\to \infty$,
\begin{eqnarray*}&&
\frac{\int_0^T \left|\frac{Z'}{Z}(t)\right|^{2K-2M} |Z(t)|^{2K}~dt}{\int_0^T |Z(t)|^{2K}~dt}
 \sim \left({i}~{\log T} \frac{ d}{ dx}\right)^{2K-2M}  \exp \left( \frac x 2 -\int_0^{4x} \frac{ds}{s} (\sigma_{I\!I\!I'}(s) +K^2)\right)  \Bigg|_{x=0}
 \end{eqnarray*}
where $\sigma_{I\!I\!I'}$ is defined in \eqref{eq:painlevedef}.
\end{conjecture}
In this formulation our conjecture appears as an average of $|Z'/Z|^{2K-2M}$  measured against $|Z|^{2K}$. Notice that the arithmetic factors cancel out as well as the ratio of the product of factorials. 

 The analogue of Theorem \ref{thm:2} is best expressed in terms of moments of the logarithmic derivative of $\zeta$. 
 In the work of Conrey and Snaith \cite{kn:consna08} on the n-correlation of the zeros of the Riemann zeta-function the authors use the ``recipe''
 to find a conjectural formula for the average over $t$ of the  product of any number of factors of the form 
 $\frac{\zeta'}{\zeta}(1/2\pm it +\alpha)$ with different values of $\alpha$ . In this work we are focused on the behavior
 of the absolute value of such a product when all of the $\alpha$ are the same and $\alpha\to 0^+$. 
 \begin{conjecture}
For any positive  integer $K$  and $a>0$ we conjecture that 
\begin{eqnarray*}
\lim_{T\to \infty} \frac 1 {T (\log T)^{2K}}\int_0^T \left| \frac{\zeta'}{\zeta}\left(\frac 12 +\frac{a}{\log T}+it\right)\right|^{2K} ~dt
=c(a)\end{eqnarray*}
where
$$\lim_{a \to 0^+} c(a)(2a)^{2K-1} =\binom{2K-2}{K-1}.$$
\end{conjecture}

One starting point for averages of characteristic polynomials and derivatives are the ratios formulae of Conrey, Farmer and Zirnbauer \cite{kn:cfz1,kn:cfz2}.
There are two useful representations of the ratios theorems.  The first is as a multiple integral:

\begin{theorem} (Conrey, Farmer, Zirnbauer \cite{kn:cfz2}) \label{theo:integralversion}
Suppose $N\geq \max\{Q-K,R-L\}$ and $\Re \gamma_q,\Re\delta_r>0$. Then
 \begin{align*}
 &\int_{U(N)}\frac{\prod_{j=1}^K\Lambda_X(e^{-\alpha_j})\prod_{l=K+1}^{K+L}\Lambda_{X^*}(e^{\alpha_l})}{\prod_{q=1}^Q\Lambda_X(e^{-\gamma_q})\prod_{r=1}^R\Lambda_{X^*}(e^{-\delta_r})}dX_N\\
 &=e^{N/2\left(-\sum_{k=1}^K\alpha_k+\sum_{l=1}^L\alpha_{K+l}\right)}\frac{(-1)^{(K+L)(K+L-1)/2}}{K!L!(2\pi i)^{K+L}}\\
 &\quad \times \oint \cdots \oint e^{N/2\left(\sum_{k=1}^Kw_k-\sum_{l=1}^Lw_{K+l}\right)}\frac{\prod_{j=1}^K\prod_{l=1}^Lz(w_j-w_{K+l})\prod_{q=1}^Q\prod_{r=1}^Rz(\gamma_q+\delta_r)}{\prod_{j=1}^K\prod_{r=1}^Rz(w_j+\delta_r)\prod_{l=1}^L\prod_{q=1}^Qz(-w_{K+l}+\delta_q)}\\
 &\qquad\qquad\qquad\qquad \times\frac{\Delta(w_1,\dots,w_{K+L})^2\prod_{j=1}^{K+L}dw_j}{\prod_{j=1}^{K+L}\prod_{k=1}^{K+L}(w_k-\alpha_j)},
 \end{align*}
where the $w$ contours enclose the $\alpha$s.
\end{theorem}
Here and in the rest of the paper, 
\begin{equation}\label{eq:vandermonde} 
\Delta(w_1,\ldots,w_n)=\prod_{1\leq j<k\leq n}(w_k-w_j)=\det\left(w_i^{j-1}\right)_{i,j=1,\ldots,n}
\end{equation}
is a Vandermonde determinant, and 
\begin{equation}\label{5.z}
z(x) := \frac{1}{1-e^{-x}}.
\end{equation}

The quantity in Theorem \ref{theo:integralversion} can also be written as a permutation sum:

\begin{theorem}\label{p1}(Conrey, Farmer, Zirnbauer \cite{kn:cfz2}, see also \cite{kn:cfs05})
Suppose $N\geq \max\{Q-K,R-L\}$ and $\Re\gamma_q,\Re\delta_r>0$. We have
\begin{eqnarray}\label{5.z1}
\lefteqn{
\int_{U(N)}
{\prod_{k=1}^K \Lambda_X(e^{-\alpha_k})
\prod_{l=K+1}^{K+L} \Lambda_{X^*}(e^{\alpha_l})  \over
\prod_{q=1}^Q \Lambda_X(e^{-\gamma_q}) \prod_{r=1}^R
\Lambda_{X^*}(e^{-\delta_r}) } dX_N
 } \nonumber \\
&& \quad  = \sum_{\sigma \in \Xi_{K,L} } e^{N \sum_{k=1}^K (\alpha_{\sigma(k)} -
\alpha_k) }
{\prod_{k=1}^K \prod_{l=K+1}^{K+L} z(\alpha_{\sigma(k)} -
\alpha_{\sigma(l)} ) \prod_{q=1}^Q \prod_{r=1}^R z(\gamma_q+\delta_r)
\over
\prod_{r=1}^R \prod_{k=1}^K z(\alpha_{\sigma(k)} + \delta_r)
\prod_{q=1}^Q \prod_{l=K+1}^{K+L} z(\gamma_q - \alpha_{\sigma(l)}) }.
\end{eqnarray}
\end{theorem}
Above, $\Xi_{K,L}$ denotes the set of
permutations $\sigma$ of $\{1,2,\dots,K+L\}$ such that
\begin{equation}
1 \le \sigma(1) < \cdots <  \sigma(K) \le K+L \quad {\rm and} \quad
1 \le \sigma(K+1) < \cdots <  \sigma(K+L) \le K+L.
\end{equation}

We are interested in the moments of logarithmic derivatives of the characteristic polynomial, which can be derived by differentiation of the ratios theorem in the form of Theorem \ref{p1}.  In the current paper we will focus on the leading order contribution to averages of the logarithmic derivatives when $N$, the matrix size, is large and the characteristic polynomial is evaluated close to the unit circle.  The exact formula for these moments of the logarithmic derivative is stated below, but it is complicated to work with and extracting the leading order behaviour is difficult (see Section \ref{sect:comparison}).  Note that the theorem below uses set rather than permutation notation for the arguments of the characteristic polynomials.  The result is merely a differentiation of Theorem \ref{p1} but given the difficulty of keeping track of all the terms, the discussion of the proof in \cite{kn:consna08} may be useful. 
\begin{theorem} (Conrey and Snaith \cite{kn:consna08}) \label{theo:logderivexact}
\label{theo:J} If $\Re \alpha_j>0$ and $ \Re \beta_j
>0$ for  $\alpha_j\in A$ and $\beta_j\in B$, then $J(A;B)=J^*(A;B)$  where
\begin{eqnarray}
J(A;B)&:=& \int_{U(N)}\prod_{\alpha\in A}
(-e^{-\alpha})\frac{\Lambda_X'}{\Lambda_X}(e^{-\alpha})\prod_{\beta\in
B} (-e^{-\beta})\frac{\Lambda_{X^*}'}{\Lambda_{X^*}}(e^{-\beta})~ dX_N
,
\end{eqnarray}
 \begin{eqnarray} &&J^*(A;B):= \nonumber \\
 &&\qquad\qquad\sum_{S\subset A,T\subset B\atop
|S|=|T|}e^{-N(\sum_{\hat \alpha\in S} \hat \alpha
+\sum_{\hat{\beta}\in T}\hat\beta)} \frac{Z(S,T)Z(S^-,T^-)} {
Z^{\dagger}(S,S^-)Z^{\dagger}(T,T^-)} \sum_{{(A-S)+ (B-T)\atop =
U_1+\dots + U_R}\atop |U_r|\le 2}\prod_{r=1}^R H_{S,T}(U_r),
\end{eqnarray}
and
\begin{equation}\label{eqn:Hrmt}
H_{S,T}(W)=\left\{\begin{array}{ll} \sum_{\hat \alpha\in
S}\frac{z'}{z}(\alpha-\hat{\alpha})-\sum_{\hat\beta\in T}
\frac{z'}{z}(\alpha +\hat \beta) &\mbox{ if $W=\{\alpha\}\subset
A-S$}
   \\
\sum_{\hat\beta\in T}\frac{z'}{z}(\beta-\hat
\beta)-\sum_{\hat\alpha\in S} \frac{z'}{z}
(\beta+\hat\alpha) &\mbox{ if  $W=\{\beta\}\subset B-T$}\\
\left(\frac{z'}{z}\right)'(\alpha+\beta) & \mbox{ if
$W=\{\alpha,\beta\}$ with $
{\alpha \in A-S, \atop \beta\in B-T}$}\\
0&\mbox{ otherwise}.
\end{array}\right.
\end{equation}
Here $z(x)=(1-e^{-x})^{-1}$, $S^-=\{-s:s\in S\}$ (similarly for $T^-$) and  $Z(A,B)=\prod_{\alpha\in A\atop\beta\in B}z(\alpha+\beta)$,
with the dagger on $Z^\dagger(S,S^-)$ imposing the additional
restriction that a factor $z(x)$ is omitted if its argument is
zero.
\end{theorem}

%%%%%%%%%%%%%%%%%%%%%%%%%%%%%%%%%%%%%%%%%%%%%%%%%%%%%%%%%%%%%%%%%%%%%%%%

\section{Proof of Theorem \ref{thm:1}}

\label{sect:theorem1}
There are various ways to write moments of the function $Z_X(s)$, defined in (\ref{eq:Z}). For example, there is an expression as a permutation sum:
\begin{eqnarray} \label{eq:summoment}
&&\int_{U(N)} \prod_{j=1}^K Z_X(e^{-\alpha_j}) Z_{X^*}(e^{\alpha_{j+K}}) dX_N  \\
&&\quad=(-1)^{NK} e^{-\tfrac{N}{2}\sum_{j=1}^{2K} \alpha_j} \sum_{\sigma\in \Xi} e^{N\sum_{j=1}^K\alpha_{\sigma(j)}} \prod_{\substack{1\leq i\leq K\\1\leq j\leq K}} z(\alpha_{\sigma(j)}-\alpha_{\sigma(K+j)}) \nonumber, 
\end{eqnarray}
where 
\begin{equation} \label{eq:z}
z(x)=\frac{1}{1-e^{-x}}=\frac{1}{x} +O(1), \; {\rm for\;small\;}x,
\end{equation} 
and $\Xi$  denotes the subset of permutations $\sigma\in S_{2K}$ of $\{1,2,\ldots,2K\}$ for which
\begin{equation}
\sigma(1)<\sigma(2)<\cdots <\sigma(K)\;\;{\rm and}\;\;
\sigma(K+1)<\sigma(K+2)<\cdots <\sigma(2K).
\end{equation}
Equation (\ref{eq:summoment}) is just a simple case of Theorem \ref{p1}, with a different prefactor because we are using $Z_X$ instead of $\Lambda_X$. This can equivalently be written as
\begin{eqnarray} \label{eq:kfoldintegral}
&&\int_{U(N)} \prod_{j=1}^K Z_X(e^{-\alpha_j}) Z_{X^*}(e^{\alpha_{j+K}}) dX_N =(-1)^{NK+K(K-1)/2} \;\frac{e^{-\tfrac{N}{2} \sum_{j=1}^{2K} \alpha_j}}{K! (2\pi i)^K}\label{eq:intmoment} \\
&&\quad \times\oint\cdots \oint  e^{N\sum_{i=1}^K w_i}\prod_{\substack{1\leq i\leq K\\1\leq j\leq 2K}} z(w_i-\alpha_j)\Delta(w_1,\ldots,w_K)^2 dw_1\cdots dw_K,\nonumber
\end{eqnarray}
where the contours enclose the $\alpha$'s, because an evaluation of residues in this integral yields the sum (\ref{eq:summoment}). 
(For more explanation of these expressions, see Section 3 of \cite{kn:crs06}, which draws on Section 2 of \cite{kn:cfkrs1}.)  Equation (\ref{eq:kfoldintegral}) is similar in spirit to Theorem \ref{theo:integralversion} except that in this simpler case the average can be written as a $K$-fold integral rather than a $K+L$ dimensional integral as in the theorem. 

We are interested in the average
\begin{equation}
\int_{U(N)} |Z_X^\prime(1)|^{2K-2M}|Z_X(1)|^{2M}dX_N.
\end{equation} 
We set $M$ and $K$ to be integers, with $2M \geq 0$ and $2K\geq 2M$, and  follow closely the method of Conrey, Rubinstein, and Snaith \cite{kn:crs06}. 
We take $2K-2M$ derivatives of  (\ref{eq:intmoment}) by applying
\begin{equation}
    \prod_{j=1}^{K-M}\frac{d}{d\alpha_j}\frac{d}{d\alpha_{j+K}},
\end{equation}
and note that
\begin{equation}
\frac{d}{d\alpha} Z_X(e^{-\alpha}) \Big|_{\alpha=0} = -e^{-\alpha} Z_X^\prime (e^{-\alpha})\Big|_{\alpha=0}=-Z_X^\prime (1)
\end{equation}
and
\begin{equation}
\frac{d}{d\alpha}Z_{X^*}(e^\alpha)\Big|_{\alpha=0} = (-1)^N\overline{Z_X^\prime (1)}.
\end{equation}
%the factor of (-1)^N comes from the factor exp(-pi i N/2), which when conjugated has the wrong sign, so we need to correct
%with a factor of exp (-pi i N) = (-1)^N
Combining these we have 
\begin{eqnarray}
&&\int_{U(N)}|Z_X^\prime(1)|^{2K-2M}|Z_X(1)|^{2M}dX_N=(-1)^{K(K+1)/2-M}\prod_{j=1}^{K-M}\frac{d}{d\alpha_j}\frac{d}{d\alpha_{j+K}}\;\frac{e^{-\tfrac{N}{2}\sum_{j=1}^{2K}\alpha_j}}{K!(2\pi i)^K}\nonumber\\
&&\qquad \times\oint\cdots \oint e^{N\sum_{i=1}^k w_i} \prod_{\substack{1\leq i\leq K\\1\leq j\leq 2K}} z(w_i-\alpha_j)\Delta^2(w_1,\ldots,w_K)dw_1\dots dw_K\Big|_{\alpha_1=\cdots=\alpha_{2K}=0}.
\label{eq:Z'Z}
\end{eqnarray}

Let $\alpha_i=a_i/N$ and $w_i\rightarrow w_i/N$, then with the use of (\ref{eq:z}) we find
\begin{eqnarray}
&&\int_{U(N)}|Z_X^\prime(1)|^{2K-2M}|Z_X(1)|^{2M}dX_N=(-1)^{K(K+1)/2-M}N^{2K-2M}N^{K^2}\prod_{j=1}^{K-M}\frac{d}{da_j}\frac{d}{da_{j+K}}\;\frac{e^{-\tfrac{1}{2}\sum_{j=1}^{2K}a_j}}{K!(2\pi i)^K}\nonumber\\
&&\qquad \times\oint\cdots \oint  e^{\sum_{i=1}^kw_i}\frac{\Delta^2(w_1,\ldots,w_K)}{\prod_{\substack{1<i<K\\1\leq j\leq 2K}}(w_i-a_j)}\big(1+O\big(\tfrac{1}{N}\big)\big)dw_1\dots dw_K\Big|_{a_1=\cdots=a_{2K}=0},
\label{eq:Z'Z2}
\end{eqnarray}
where the contours enclose the $a$'s.

The aim now is to separate these integrals. We do this by using a series of results from \cite{kn:crs06}.  
%For a start we have (\cite{kn:crs06}, equation (4.6))
To start,

\begin{equation}
\left.\frac{d}{da}\frac{e^{-\tfrac{a}{2}}}{\prod_{1\leq i\leq k}(w_i-a)}\right|_{a=0}=\frac{1}{\prod_{i=1}^kw_i}\left(\sum_{j=1}^k\frac{1}{w_j}-\frac{1}{2}\right).
\end{equation}

Next we allow $\Delta\left( \frac{d}{dL}\right)$ to have the meaning
\begin{equation}
\Delta\left( \frac{d}{dL}\right) \prod_{i=1}^K f(L_i)=\prod_{1\leq i<j \leq K} \left( \frac{d}{dL_j}-\frac{d}{dL_i}\right) \prod_{i=1}^Kf(L_i).
\end{equation}
Below, we will use Lemma 5 of \cite{kn:crs06}:
\begin{equation} \label{eq:delta2}
\Delta^2 \left( \frac{d}{dL}\right) \left( \prod_{i=1}^K f(L_i)\right) \Big|_{L_i=1} =K! \det_{K\times K} \left(f^{(i+j-2)}(1)\right),
\end{equation}
for any sufficiently differentiable function $f$.

Borrowing a technique from~\cite{kn:crs06}, we replace the factor $\exp( \sum w_i)$ appearing in~\eqref{eq:Z'Z2}  by $\exp(\sum L_i w_i)$,
and then pull out the Vandermonde determinant squared from the integrand as a differential operator.
Differentiating under the integral sign and substituting $L_i=1$ then recovers the original integral.
The advantage in doing so is that it allows us to separate the resulting multidimensional integral.

Thus, we have
\begin{eqnarray}
&&\int_{U(N)}|Z_X^\prime(1)|^{2K-2M}|Z_X(1)|^{2M}dX_N=(-1)^{K(K+1)/2-M} N^{K^2+2K-2M}\frac{\Delta^2\left(\frac{d}{dL}\right)}{K!(2\pi i)^K} \\
&&\qquad \times \oint\cdots \oint \frac{e^{\sum_{i=1}^K L_iw_i} \left( \sum_{j=1}^K \frac{1}{w_j} -\frac{1}{2}\right)^{2K-2M}} {\prod_{i=1}^K w_i^{2K}}\big(1+O\big(\tfrac{1}{N}\big)\big) dw_1 \cdots dw_K \Big|_{L_i=1} \nonumber \\
&& \quad=(-1)^{K(K+1)/2-M}N^{K^2+2K-2M}\frac{\Delta^2\left(\frac{d}{d L}\right)}{K!}\nonumber \\
&&\qquad\qquad \times\left(\frac{d}{dx}\right)^{2K-2M}e^{-\frac{x}{2}} \prod_{j=1}^K \left( \frac{1}{2\pi i}  \oint\frac{e^{L_jw+x/w}}{w^{2K}}dw\right)\big(1+O\big(\tfrac{1}{N}\big)\big) \Big|_{L_j=1,x=0}.\nonumber
\end{eqnarray}
In the last line we have again introduced an extra parameter $x$ and introduced the differential operator in the $x$ variable so as
to simplify the $(\sum1/w_j-1/2)^{2K-2M}$ appearing in the integrand.

Still following \cite{kn:crs06} we have, from equation (2.11) of that paper, 
\begin{equation}
\frac{1}{2\pi i} \int_{|z|=1} \frac{e^{Lz+t/z}}{z^{2k}}dz=\frac{L^{2k-1} I_{2k-1}(2\sqrt{Lt})}{(Lt)^{k-1/2}} =:f_t(L),
\end{equation}
where $I$ is the I-Bessel function. We now use (\ref{eq:delta2}) and write
\begin{eqnarray}
&&\int_{U(N)}|Z_X^\prime(1)|^{2K-2M}|Z_X(1)|^{2M}dX_N=(-1)^{K(K-1)/2+K-M} N^{K^2+2K-2M}\\
&&\qquad \times \left(\frac{d}{dx}\right)^{2K-2M}e^{-\frac{x}{2}} \det_{K\times K} \left( f_x^{(i+j-2)}(1)\right) \big(1+O\big(\tfrac{1}{N}\big)\big)\Big|_{x=0} .\nonumber
\end{eqnarray}
Using (4.15) of \cite{kn:crs06}, and performing some manipulations on the determinant,  we end up with equation~\eqref{eq:generalization of CRS}
of Theorem~\ref{thm:1}.
%\begin{eqnarray}
%\notag
%&&\int_{U(N)}|Z_A^\prime(1)|^{2K-2M}|Z_A(1)|^{2M}dA_N=(-1)^{K(K-1)/2+K-M} \\
%&&\qquad \times N^{K^2+2K-2M} \left( \frac{d}{dx}\right) ^{2K-2M} \left( e^{-x/2}x^{-K^2/2} \det_{K\times K}\left( I_{i+j-1}(2\sqrt{x})\right)\right) \Bigg|_{x=0}.
%\label{eq:generalization of CRS}
%\end{eqnarray}

Finally, from \cite{kn:forwit06} we have that 
\begin{eqnarray}\label{p3}
&&\exp\left(-\int_0^{4x} \frac{ds}{s} (\sigma_{I\!I\!I'}(s)+k^2)\right) =(-1)^{k(k-1)/2}\nonumber \\
&&\qquad \times \prod_{j=0} ^{k-1} \frac{(j+k)!}{j!} x^{-k^2/2} e^{-x} \det_{k\times k} \left(I_{i+j-1}(2\sqrt{x})\right), 
\end{eqnarray}
where $\sigma_{I\!I\!I'}(s)$ is the solution of the Painlev{\'e} equation
\begin{equation}\label{eq:painlevedef}
(s\; \sigma_{I\!I\!I'}'')^2+\sigma_{I\!I\!I'}'(4\sigma_{I\!I\!I'}'-1)(\sigma_{I\!I\!I'}-s\;\sigma_{I\!I\!I'}') -\frac{k^2}{16}=0,
\end{equation}
satisfying the boundary condition
\begin{equation}
\sigma_{I\!I\!I'}(s) \widesim{s\rightarrow0}-k^2 +\frac{s}{8} + O(s^2),
\end{equation}
for $k\in\mathbb{N}$. This means that the average we are looking at is related to the solution of the Painlev{\'e} equation in the following manner
\begin{eqnarray}
&&\int_{U(N)}|Z_X^\prime(1)|^{2K-2M}|Z_X(1)|^{2M}dX_N=(-1)^{K-M} N^{K^2+2K-2M}\nonumber \\
&& \qquad \times \prod_{j=0}^{K-1} \frac{j!}{(j+K)!} \left( \frac{d}{dx} \right) ^{2K-2M} e^{x/2} \exp \left( -\int_0^{4x} \frac{ds}{s} (\sigma_{I\!I\!I'}(s) +K^2)\right)\big(1+O\big(\tfrac{1}{N}\big)\big) \Bigg|_{x=0}.
\label{eq:painleve}
\end{eqnarray}
This reduces to 
\begin{equation}
N^{K^2} \prod_{j=0}^{K-1} \frac{j!}{(j+K)!}
\end{equation}
if $M=K$, which is the familiar result for the basic moment. 

%Want to try with h=\mathbb{N}+\frac{1}{2} and k\in\mathbb{N}, firstly with k=1, h=1/2

%%%%%%%%%%%%%%%%%%%%%%%%%%%%%%%%%%%%%%%%%%%%%%%%%%%%%%%%%%%%%%%%%%%%%%%%%%%%%%

%now looking at logarithmic derivative
\section{Proof of Theorem~\ref{thm:2}}
\label{sect:theorem2}

We first note that
\begin{equation}
    \frac{d}{d\alpha}\Lambda_X(e^{-\alpha})
     =-\Lambda^\prime_X(e^{-\alpha})e^{-\alpha},
\end{equation}
and
\begin{equation}
\frac{d}{d\alpha}\Lambda_{X^*}(e^\alpha)=\Lambda_{X^*}^\prime(e^\alpha)e^\alpha.
\end{equation}

We thus find, on setting $L=Q=R=K$ in Theorem~\ref{theo:integralversion}, then differentiating with respect to all the
$\alpha$'s, and subsequently setting  
 $\alpha_1=\dots=\alpha_K=\alpha$, $\alpha_{K+1}=\dots=\alpha_{2K}=-\alpha$, and all $\gamma,\delta=\alpha$, with $\Re \alpha>0$,
\begin{align}
 \notag
(-1)^K&\int_{U(N)}\left(\frac{\Lambda_X^\prime(e^{-\alpha})}{\Lambda_X(e^{-\alpha})}\frac{\Lambda_{X^*}^\prime(e^{-\alpha})}{\Lambda_{X^*}(e^{-\alpha})}\right)^Ke^{-2\alpha K}dX_N\\
\notag
 &=\prod_{j=1}^{2K}\frac{d}{d\alpha_j}\left(e^{N/2\left(-\sum_{j=1}^K\alpha_j+\sum_{j=K+1}^{2K}\alpha_j\right)}\frac{(-1)^{K(2K-1)}}{K!^2(2\pi i)^{2K}}\right.\\
 \notag
 &\times \oint\cdots \oint e^{N/2\left(\sum_{j=1}^Kw_j-\sum_{l=1}^Kw_{K+l}\right)}\frac{\prod_{j=1}^K\prod_{l=1}^Kz(w_j-w_{l+K})\prod_{q=1}^K\prod_{r=1}^Kz(\gamma_q+\delta_r)}{\prod_{j=1}^K\prod_{r=1}^Kz(w_j+\delta_r)\prod_{l=1}^K\prod_{q=1}^Kz(-w_{K+l}+\delta_q)}\\
 &\times \left.\frac{\Delta(w_1,\dots,w_{2K})^2}{\prod_{j=1}^{2K}\prod_{k=1}^{2K}(w_k-\alpha_j)}dw_1\dots dw_{2K}\right)_{\substack{\alpha_1=\dots=\alpha_K=\alpha\\\alpha_{K+1}=\dots=\alpha_{2K}=-\alpha\\\text{all }\gamma,\delta=\alpha}}.
\label{eq:biggie}
 \end{align}
 
 Note that \[\frac{d}{d\alpha}\frac{e^{\pm N\alpha/2}}{\prod_j(w_j-\alpha)}%=\pm\frac{N}{2}e^{\pm N\alpha/2}+\frac{e^{\pm N\alpha/2}}{\prod_j(z_j-\alpha)}\sum_j\frac{1}{z_j-\alpha}
 =\frac{e^{\pm N\alpha/2}}{\prod_j(w_j-\alpha)}\left(\sum_j\frac{1}{(w_j-\alpha)}\pm \frac{N}{2}\right),\] so we perform the derivative and then the substitutions of the $\alpha_j$, $\gamma_j$ and $\delta_j$ to obtain
 
 \begin{align}
 \notag
(-1)^K&\int_{U(N)}\left(\frac{\Lambda_X^\prime(e^{-\alpha})}{\Lambda_X(e^{-\alpha})}\frac{\Lambda_{X^*}^\prime(e^{-\alpha})}{\Lambda_{X^*}(e^{-\alpha})}\right)^Ke^{-2\alpha K}dX_N\\
\notag
 &=\frac{(-1)^{K}}{K!^2(2\pi i)^{2K}}e^{-NK\alpha}z(2\alpha)^{K^2}\\
 \notag
 &\times \oint\cdots \oint e^{N/2\left(\sum_{j=1}^Kw_j-\sum_{l=1}^Kw_{K+l}\right)}\frac{\prod_{j=1}^K\prod_{l=1}^Kz(w_j-w_{K+l})}{\prod_{j=1}^Kz(w_j+\alpha)^K\prod_{l=1}^Kz(-w_{K+l}+\alpha)^K}\\
 &\times \frac{\Delta(w_1,\dots,w_{2K})^2\big( \sum_{j=1}^{2K} \frac{1}{w_j-\alpha}-\frac{N}{2}\big) ^K \big( \sum_{j=1}^{2K} \frac{1}{w_j+\alpha}+\frac{N}{2}\big)^K}{\prod_{j=1}^{2K}(w_j-\alpha)^K(w_j+\alpha)^K}dw_1\dots dw_{2K}.
\label{eq:biggie1b}
 \end{align}

To obtain the leading order asymptotics in $N$, let $\alpha=a/N$, with $a=o(1)$ as $N \to \infty$, and substitute $w_j=au_j/N$.  For large $N$, we can now simplify the integrand by
replacing each occurrence of the function $z(x)$ by $1/x$. The double product involving $z(w_j-w_{l+K})$ thus cancels a portion of the
$\Delta(w_1,\dots,w_{2K})^2$, up to a factor of $(-1)^{K^2}$, and so we let
\begin{equation}
\label{eq:triple vandermonde}
    q(w_1,\ldots, w_{2K})=
    \Delta( w_1,\ldots, w_{2K})
    \Delta( w_1,\ldots, w_{K})
    \Delta( w_{K+1},\ldots, w_{2K}).
\end{equation}
As in the previous section, we can introduce extra variables $L_j$ and pull out the polynomial $q$ from the integrand
as a differential operator. Similarly, the factors containing $z(x)$ in the denominator cancel some of the $(w_j-\alpha)(w_j+\alpha)$ factors in the denominator, again up to a $(-1)^{K^2}$. 
Carrying out these steps, and cancelling out the powers of $a$ that can be pulled outside the integral,~\eqref{eq:biggie1b} becomes
\begin{align}
\notag
% the power of -1: all the -1 from the previous display comes out as an even power. However, there is an extra (-1)^{K^2} from
% cancelling K^2 factors of the Vandermonde with the opposite sign for each factor
(-1)^K&\int_{U(N)}\left(\frac{\Lambda_X^\prime(e^{-a/N})}{\Lambda_X(e^{-a/N})}\frac{\Lambda_{X^*}^\prime(e^{-a/N})}{\Lambda_{X^*}(e^{-a/N})}\right)^Ke^{-2a K/N}dX_N\\ \notag
=\frac{(-1)^{K}}{K!^2 (2\pi i)^{2K}}&\frac{e^{-aK}N^{2K}}{2^{K^2}} \, q\left(\frac{d}{dL}\right)\oint\cdots \oint e^{\sum_{j=1}^{2K}u_jL_j}\\
&\times\frac{\left(\sum_{j=1}^{2K}\frac{1}{au_j-a}-\frac{1}{2}\right)^K\left(\sum_{j=1}^{2K}\frac{1}{au_j+a}+\frac{1}{2}\right)^K}{\prod_{j=1}^{K}(u_j-1)^K\prod_{j=K+1}^{2K}(u_j+1)^K}du_1\dots du_{2K} 
\Big|_{\substack{L_1,\ldots,L_K=a/2\\L_{K+1},\ldots,L_{2K}=-a/2}}\times\left(1+O(\tfrac{a}{N})\right),
\end{align}
where the contours of integration enclose $\pm 1$. 

Introducing more variables $t_1,t_2$, the above can be written as
\begin{multline}
    \frac{(-1)^Ke^{-aK}}{K!^2 (2\pi i)^{2K} 2^{K^2}}\left(\frac{N}{a}\right)^{2K}
    \left(\frac{d}{dt_1}\right)^K\left(\frac{d}{dt_2}\right)^K
    e^{a(t_2-t_1)/2}
    q\left(\frac{d}{d L}\right) \\
   \times \oint\cdots \oint
    \frac{\exp\left(\sum_{j=1}^{2K}u_jL_j+\frac{t_1}{u_j-1}+\frac{t_2}{u_j+1}\right)}
    {\prod_{j=1}^{K}(u_j-1)^K\prod_{j=K+1}^{2K}(u_j+1)^K}du_1\dots du_{2K}
    \Big|_{\substack{L_1,\ldots,L_K=a/2\\L_{K+1},\ldots,L_{2K}=-a/2\\t_1,t_2=0}}\times\left(1+O(\tfrac{a}{N})\right),
    \label{eq:biggie2}
\end{multline}
where the contours encircle $\pm 1$. Now, the $2K$ dimensional residue above can be separated into a product:
\begin{equation}
    \frac{1}{(2\pi i)^{2K}}
    \oint\cdots \oint
    \frac{\exp\left(\sum_{j=1}^{2K}u_jL_j+\frac{t_1}{u_j-1}+\frac{t_2}{u_j+1}\right)}
    {\prod_{j=1}^{K}(u_j-1)^K\prod_{j=K+1}^{2K}(u_j+1)^K}du_1\dots du_{2K}
    = \prod_{j=1}^Kf(L_j)\prod_{j=K+1}^{2K}g(L_j),
\label{eq:separate residue}
\end{equation}
where
\begin{equation}
    f(L)=\frac{1}{2\pi i}\oint\frac{\exp\left(uL+\frac{t_1}{u-1}+\frac{t_2}{u+1}\right)}{(u-1)^K}du,
\end{equation}
and
\begin{equation}
    g(L)=\frac{1}{2\pi i}\oint\frac{\exp\left(uL+\frac{t_1}{u-1}+\frac{t_2}{u+1}\right)}{(u+1)^K}du.
\end{equation}
The technique of Lemma 2.2 from ~\cite{kn:cfkrs2} can be adapted, and~\eqref{eq:biggie2} becomes
\begin{eqnarray}\label{eq:bhd}
&&(-1)^K\int_{U(N)}\left(\frac{\Lambda_X^\prime(e^{-a/N})}{\Lambda_X(e^{-a/N})}\frac{\Lambda_{X^*}^\prime(e^{-a/N})}{\Lambda_{X^*}(e^{-a/N})}\right)^Ke^{-2a K/N}dX_N\notag \\
&&\qquad=    \frac{(-1)^K e^{-aK}}{2^{K^2}} \left(\frac{N}{a}\right)^{2K}
    \left(\frac{d}{dt_1}\right)^K\left(\frac{d}{dt_2}\right)^K
    e^{a(t_2-t_1)/2}\\
    \notag&& \qquad \qquad\times
    \det_{2K\times 2K}\begin{pmatrix}f^{(i+j-2)}(a/2)\\g^{(i+j-2)}(-a/2)\end{pmatrix}
    \Big|_{t_1,t_2=0}\times\left(1+O(\tfrac{a}{N})\right)
    \label{eq:biggie3}
\end{eqnarray}
where the first $K$ rows of the above matrix ($1\leq i\leq K$) have entries $f^{(i+j-2)}(a/2)$, in column $1 \leq j \leq 2K$,
and the last $K$ rows (rows $i+K$, with $1\leq i\leq K$) have entries $g^{(i+j-2)}(-a/2)$, in column $1 \leq j \leq 2K$. We study the determinant obtained here from the point of view of Riemann-Hilbert problems in  Section \ref{sect:rhp}.

Next, we drop the factors $e^{a(t_2-t_1)/2}$ and $e^{-aK}$ as they do not affect the asymptotic since $a$ becomes small as $N\rightarrow\infty$.
We also approximate, in the integrands, $\exp(\pm au/2)$ by $1 \pm au/2$. One might think
that to obtain just the leading order asymptotic for small $a$, the $\pm au/2$ would not be needed. However, this turns
out to be incorrect, since if we just approximate by 1, i.e. without the term $\pm au/2$, the resulting
determinant is magically
independent of $t_1$ and $t_2$, and does not survive the differentiation (Lemma~\ref{lem:1}) below.

Thus, the leading order term can be simplified to 
\begin{equation}
    \frac{(-1)^{K}}{2^{K^2}}
    \left( \frac{N}{a} \right) ^{2K}
    \left(\frac{d}{dt_1}\right)^K\left(\frac{d}{dt_2}\right)^K
    \det_{2K\times 2K} M_2
    \Big|_{t_1,t_2=0},
    \label{eq:getting there}
\end{equation}
where the matrix $M_2$ has entries in the top $K$ rows of
\begin{equation}
    \label{eq:M_2 upper}
    \frac{1}{2\pi i}\oint\frac{(1+au/2) u^{i+j-2}\exp\left(\frac{t_1}{u-1}+\frac{t_2}{u+1}\right)}{(u-1)^K}du,
\end{equation}
and entries in the final $K$ rows of
\begin{equation}
    \label{eq:M_2 lower}
    \frac{1}{2\pi i}\oint\frac{(1-au/2) u^{i+j-2}\exp\left(\frac{t_1}{u-1}+\frac{t_2}{u+1}\right)}{(u+1)^K}du.
\end{equation}

Returning to~\eqref{eq:biggie3}, dropping the $\exp(-2 a K/N)$ as it
does not impact the leading asymptotic when $N$ is large, we have determined,
\begin{equation}
   \int_{U(N)}\left|\frac{\Lambda_X^\prime}{\Lambda_X}(e^{-a/N})\right|^{2K}dX_N
   =
    \frac{1}{2^{K^2}}
    \left( \frac{N}{a} \right) ^{2K}
    \left(\frac{d}{dt_1}\right)^K\left(\frac{d}{dt_2}\right)^K
    \det_{2K\times 2K} M_2
    \Big|_{t_1,t_2=0}\times(1+O(a))
    \label{eq:looking good}
\end{equation}
as $N \to \infty$ with $a=\alpha N \to 0$.

Consider now the factor $1 \pm au/2$ that appears in the entries~\eqref{eq:M_2 upper} and~\eqref{eq:M_2 lower}.
The role of this factor can be analyzed using the following property of determinants:
let $A$ be an $n\times n$ matrix, let $a_1,\dots,a_n$ denote the rows (or columns) of $A$, and let $v$ be an $n$-dimensional vector.  Then for any scalar $x$,
\begin{equation}
    \label{detrule}\det(a_1,\dots,a_j+x v,\dots,a_n)=\det(A)+x\det(a_1,\dots,v,\dots,a_n).
\end{equation}
Expanding in this fashion, the $1 \pm au/2$ results in two determinants for each row, so $2^{2K}$ determinants
altogether.

The lemma below describes what happens in the simplest of cases, where we select, for each row, just the 1
from $1 \pm au/2$. The proof, along with that of Lemma~\ref{lem:2} will be given in the next section.

\begin{lemma}
\label{lem:1}
We have
\begin{equation}
    \label{eq:simple matrix}
    \det_{2K\times 2K}
    \begin{pmatrix}
        \frac{1}{2\pi i}\oint\frac{ u^{i+j-2}\exp\left(\frac{t_1}{u-1}+\frac{t_2}{u+1}\right)}{(u-1)^K}du
        \\
        \frac{1}{2\pi i}\oint\frac{u^{i+j-2}\exp\left(\frac{t_1}{u-1}+\frac{t_2}{u+1}\right)}{(u+1)^K}du
    \end{pmatrix}
    =(-2)^{K^2}.
\end{equation}
\end{lemma}
Thus the determinant in the above lemma, does not depend on $t_1$ or $t_2$, and,
on applying $\left(\frac{d}{dt_1}\right)^K\left(\frac{d}{dt_2}\right)^K$, does not
contribute to~\eqref{eq:looking good}.

Now that we understand what happens when just the 1 is selected from $1\pm au/2$, we next examine the contribution from the $\pm au/2$ terms. We will only consider those determinants that
are obtained by a single selection of these terms along exactly one of the rows (as in expansion~\eqref{detrule}), as these are the determinants
that will result in the main asymptotics of size $N^{2K}/a^{2K-1}$ (each row for which we select $\pm a u/2$
increases the power of $a$ by 1, which will contribute to the asymptotic described in \eqref{eq:looking good}).

Selecting this term for each entry in a specific row has the effect of incrementing the power of $u$
in the numerator of the corresponding integrands from
$i+j-2$ to $i+j-1$. This then matches with the entries in the row below, giving a zero value for the determinant,
unless the selected row is row $K$ or $2K$. The next lemma summarizes what happens in either of these two cases.

\begin{lemma}
\label{lem:2}
Let $M_3$ be the matrix identical to the one displayed in~\eqref{eq:simple matrix}, except that the power of $u$
in the numerator of each integrand is $i+j-1$ along its $K$th row, rather than $i+j-2$. Then, $\det M_3$ is a polynomial
in $t_1$ and $t_2$ of degree $2K$, and satisfies:
\begin{equation}
    \det M_3 = \frac{K 2^{K^2-2K}}{K!^2(2K-1)} (t_1+t_2)^{2K} + O( (|t_1|+|t_2|)^{2K-1}).
\end{equation}
Furthermore, if instead of the $K$th row, we modify the $2K$th row in the same fashion, the same result holds, except
there is an extra factor of $-1$ on the right hand side of the equality.
\end{lemma}
Therefore, the contribution to~\eqref{eq:looking good} from the two determinants in the above lemma
(and the only contribution impacting the $N^{2K}/a^{2K-1}$ term) equals, on taking $a/2$ of the first determinant
and $-a/2$ of the second, 
\begin{eqnarray}
\int_{U(N)}\left|\frac{\Lambda_X^\prime}{\Lambda_X}(e^{-a/N})\right|^{2K}dX
  & =&
    \frac{1}{2^{K^2}}
    \left( \frac{N}{a} \right) ^{2K}(2K)!\frac{K\;2^{K^2-2K}}{(K!)^2 (2K-1)} \left( \frac{a}{2}-\left(-\frac{a}{2}\right)\right)\times(1+O(a))\\
    \notag &=& \binom{2K-2}{K-1} \frac{N^{2K}}{a^{2K-1}2^{2K-1}}\times(1+O(a)).
    \end{eqnarray}
 So we have (\ref{eq:final asymptotic}).

%%%%%%%%%%%%%%%%%%%%%%%%%%%%%%%%%%%%%%%%%%%%%%%%%%%%%%%%%%%%%%

\section{Proof of Lemma \ref{lem:1}}

We first establish that the determinant in the lemma is independent of $t_1$ and $t_2$ by showing that its derivative
with respect to either variable is 0.

When we differentiate with respect, say, to $t_1$ we get a sum of $2K$ determinants of the $2K$ matrices formed by
differentiating the entries of a specific column of the original matrix. We will show that each of
these $2K$ determinants is 0.

The $j$th of these determinants has
the entries of its $j$th column differentiated with respect to $t_1$, and they are equal, in the top half of
the matrix (in the $i$th row, with $1\leq i \leq K$), to
\begin{equation}
    \label{eq:top half diff}
    \frac{1}{2\pi i}\oint\frac{ u^{i+j-2}\exp\left(\frac{t_1}{u-1}+\frac{t_2}{u+1}\right)}{(u-1)^{K+1}}du
\end{equation}
and, in the bottom half (in the $(K+i)$th row, with $1\leq i\leq K$),
\begin{equation}
    \label{eq:bottom half diff}
    \frac{1}{2\pi i}\oint\frac{ u^{i+j-2}\exp\left(\frac{t_1}{u-1}+\frac{t_2}{u+1}\right)}{(u+1)^{K}(u-1)}du.
\end{equation}

If $j=1$, the integrand of each entry in this column is of size $O(|u|^{-2})$, as $|u| \to
\infty$. As $|u|\rightarrow \infty$, the length of the contour grows proportionally to $|u|$, hence taking a large contour shows that each entry in
this column is 0, and hence the determinant is 0.  

If the column being differentiated has $j>1$, we can show that the resulting column is a linear combination of
columns $1,\ldots,j-1$. For, if we add the first $j-1$ entries in the $i$th row of the top half of the matrix,
we get
\begin{equation}
    \frac{1}{2\pi i}\oint \sum_{l=1}^{j-1} \frac{ u^{i+l-2}\exp\left(\frac{t_1}{u-1}+\frac{t_2}{u+1}\right)}{(u-1)^K}du
    =\frac{1}{2\pi i}\oint \frac{ u^{i-1}(u^{j-1}-1)\exp\left(\frac{t_1}{u-1}+\frac{t_2}{u+1}\right)}{(u-1)^{K+1}}du.
\end{equation}
This nearly matches \eqref{eq:top half diff}, the difference being
\begin{equation}
    \label{eq:difference1}
    \frac{1}{2\pi i}\oint \frac{ u^{i-1}\exp\left(\frac{t_1}{u-1}+\frac{t_2}{u+1}\right)}{(u-1)^{K+1}}du.
\end{equation}
But the integrand is $O(|u|^{-2})$, hence \eqref{eq:difference1} equals 0.

Similarly, the sum of the first $j-1$ entries in row $i$ in the bottom half equals
\begin{equation}
    \frac{1}{2\pi i}\oint \sum_{l=1}^{j-1} \frac{ u^{i+l-2}\exp\left(\frac{t_1}{u-1}+\frac{t_2}{u+1}\right)}{(u+1)^K}du
    =\frac{1}{2\pi i}\oint \frac{ u^{i-1}(u^{j-1}-1)\exp\left(\frac{t_1}{u-1}+\frac{t_2}{u+1}\right)}{(u+1)^K(u-1)}du.
\end{equation}
Again, the difference between the right hand side above and \eqref{eq:bottom half diff},
\begin{equation}
    \label{eq:difference}
    \frac{1}{2\pi i}\oint \frac{ u^{i-1}\exp\left(\frac{t_1}{u-1}+\frac{t_2}{u+1}\right)}{(u+1)^K(u-1)}du,
\end{equation}
equals 0, because the integrand is $O(|u|^{-2})$.

We have thus shown that the $j$th differentiated column is equal to the sum of the first $j-1$ non-differentiated
columns, and hence the corresponding determinant is 0, as claimed.

A similar computation shows the derivative with respect to $t_2$  of the determinant in the lemma equals 0.

Having established that the left hand side of \eqref{eq:simple matrix} is independent of $t_1$ and $t_2$, we can determine
its value by specializing $t_1=t_2=0$, in which case we can evaluate the residue at $u=1$ and the top $K$ rows have entries
\begin{equation}
    \frac{1}{2\pi i}\oint\frac{ u^{i+j-2}}{(u-1)^K}du
    = \binom{i+j-2}{K-1}, \quad 1 \leq i \leq K,\, 1 \leq j \leq 2K,
\end{equation}
and the bottom $K$ rows have entries
\begin{equation}
    \frac{1}{2\pi i}\oint\frac{u^{i+j-2}}{(u+1)^K}du.
    = (-1)^{i+j-K-1} \binom{i+j-2}{K-1}, \qquad 1 \leq i \leq K,\, 1 \leq j \leq 2K.
\end{equation}
The first identity is easily obtained by writing $u^{i+j-2} =  ( (u-1) + 1)^{i+j-2}$ and extracting the coefficient of
$(u-1)^{K-1}$. The second identity follows similarly by writing $u^{i+j-2} =  ( (u+1) - 1)^{i+j-2}$.

Next, we can pull out $(-1)^{-K-1}$ from each of the bottom $K$ rows of the determinant, and as $K(K+1)$ is even,
these powers of $-1$ altogether give 1. We thus need to consider matrices of the following form: 
 \begin{equation}\label{malphazero}\mleft(\begin{array}{c} \binom{i+j-2}{K-1}\\ \scriptstyle{1\leq i\leq K, 1\leq j\leq 2K}\\\hline\\(-1)^{i+j}\binom{i+j-2}{K-1}\\ \scriptstyle{1\leq i\leq K, 1\leq j\leq 2K}
\end{array} \mright).\end{equation}
In the top half of the matrix, starting from row $K$ and working up, we subtract row $i-1$ from row $i$, $i=K,\ldots,2$ and use
Pascal's identity:
\begin{equation}
    \label{eq:pascal}
    \binom{n}{r}-\binom{n-1}{r}=\binom{n-1}{r-1}.
\end{equation}
This decreases by 1 both indices of the binomial coefficients in all elements of rows 2 to $K$ but does not change the determinant.  The first row remains unchanged. 
In the bottom half, instead of subtracting, we add row $i-1$ to row $i$ for $i=K, K-1, \ldots, 2$.

 We then repeat the procedure, but this time on rows $i=K,K-1,\ldots,3$, (this time reducing both indices of the binomial coefficients in all except the first {\em two} rows) and so on, until we
have row reduced the matrix to the following form:
\begin{equation}
    \begin{vmatrix}
    \binom{0}{K-1}&\cdots&\binom{2K-1}{K-1}\\
    \vdots&\ddots&\vdots\\
    \binom{0}{0}&\cdots&\binom{2K-1}{0}\\
    \binom{0}{K-1}&\cdots&-\binom{2K-1}{K-1}\\
    \vdots&\ddots&\vdots\\
    (-1)^{K+1}\binom{0}{0}&\cdots&(-1)^{3K}\binom{2K-1}{0}
    \end{vmatrix}.
\label{eq:finally}
\end{equation}

We now rearrange the rows.  An interchange of any two rows changes the determinant by a factor of -1. An even number of row swaps (the same for the top and bottom halves), and pulling out $(-1)^{K-1}$ from each of the $K$ bottom rows therefore does not change the determinant, but transforms it to the following form that has already been computed (see \cite{kn:cfkrs}, equation (2.7.14)):
\begin{equation}\label{papermatrix}
\begin{vmatrix}
\binom{0}{0}&\binom{1}{0}&\cdots&\binom{2K-1}{0}\\
\binom{0}{1}&\binom{1}{1}&\cdots&\binom{2K-1}{1}\\
\vdots&\vdots&\ddots&\vdots\\
\binom{0}{K-1}&\binom{1}{K-1}&\cdots&\binom{2K-1}{K-1}\\
\binom{0}{0}&-\binom{1}{0}&\cdots&-\binom{2K-1}{0}\\
-\binom{0}{1}&\binom{1}{1}&\cdots&\binom{2K-1}{1}\\
\vdots&\vdots&\ddots&\vdots\\
(-1)^{K-1}\binom{0}{K-1}&(-1)^{K}\binom{1}{K-1}&\cdots&(-1)^{K}\binom{2K-1}{K-1}\\
\end{vmatrix}=(-2)^{K^2}.\end{equation}%This needs confirming, likely can structure this better

%%%%%%%%%%%%%%%%%%%%%%%%%%%%%%%%%%%%%%%%%%%%%%%%%%%%%

\section{Proof of Lemma \ref{lem:2}}

We will first prove that $\det M_3$ is a polynomial of degree $2K$ in $t_1$ and $t_2$.
Our strategy is to show that the $(2K+1)$-st and higher partial derivatives are all 0.  This is achieved below with the help of Lemma \ref{lemma:twocolumnsequal}, Proposition \ref{prop:deg det} and Proposition \ref{prop:nmsum}.  Then in Proposition \ref{prop:coeff} we determine the value of the coefficients of the terms of order $2K$. 

Differentiating our $2K\times 2K$ determinant with respect to either variable produces, as in the proof
of the previous lemma, a sum of
$2K$ determinants where the entries of the resulting matrix are identical to the original, except that the $j$th determinant has the entries
of its $j$th column differentiated. If we repeatedly differentiate at least $2K+1$ times in total with respect to the
two $t$ variables, we get a sum of determinants, each one specified by two lists of non-negative integers
\begin{equation}
    \{m_1,\dots m_{2K}\}   \mbox{          and             }  \{n_1,\dots n_{2K}\},
\end{equation}
such that 
\begin{equation}
     m_1+\dots +m_{2K}+n_1+\dots +n_{2K}>2K.
\end{equation}
Here $m_j$ is the number of times that column $j$ has been differentiated with respect to $t_1$ and $n_j $ is the number of times column $j$ has been differentiated with respect to $t_2$.

Thus we are looking at the determinant of the matrix with upper entries
\begin{equation}\frac{1}{2\pi i}\int_{|u|=2} \frac{u^{i+j-2} \exp\left(\frac{t_1}{u-1}+\frac{t_2}{u+1}\right)}{(u-1)^{K+m_j}
(u+1)^{n_j}} ~du \qquad (1\leq i<K, 1\leq j\le 2K); \label{eq:matrixdef1}\end{equation}
\begin{equation}\frac{1}{2\pi i}\int_{|u|=2} \frac{u^{i+j-1}  \exp\left(\frac{t_1}{u-1}+\frac{t_2}{u+1}\right)}{(u-1)^{K+m_j}
(u+1)^{n_j}} ~du \qquad (i=K, 1 \leq j\le 2K );\end{equation}
and lower entries 
\begin{equation}\frac{1}{2\pi i}\int_{|u|=2} \frac{u^{i+j-2} \exp\left(\frac{t_1}{u-1}+\frac{t_2}{u+1}\right) }{(u-1)^{m_j}(u+1)^{K+n_j}} ~du \qquad (1\leq i\le K, 1\leq j\le 2K).\label{eq:matrixdef3}\end{equation}

To facilitate this discussion it is helpful to let
\begin{equation}
    I(r,E,G):=\frac{1}{2\pi i}\int_{|u|=2} \frac{u^r \exp\left(\frac{t_1}{u-1}+\frac{t_2}{u+1}\right)}  {(u-1)^{E} (u+1)^{G} }~du .
    \label{eq:I}
\end{equation}
Note that if $E+G\ge r+2$ then, as in the proof of the previous lemma,  $I(r,E,G)=0$.
Also, we have two easily proved recursion formulas:
\begin{equation}
    I(r,E,G)= I(r-1,E-1,G)+I(r-1,E,G)
\label{eq:I recursion 1}
\end{equation}
and
\begin{equation}
    I(r,E,G)= I(r-1,E,G-1)-I(r-1,E,G).
\label{eq:I recursion 2}
\end{equation}

In general we are interested in  the collection $\mathcal M=\mathcal M_{2K}$ of $2K\times 2K$ matrices $M=(M_{i,j})$ where  each entry $M_{i,j}$ is one of these integrals  
\begin{equation}M_{i,j}= I(r_{i,j}, E_j, G_j)=\frac{1}{2\pi i}\int_{|u|=2} \frac{u^{r_{i,j}}\exp\left(\frac{t_1}{u-1}+\frac{t_2}{u+1}\right)}  {(u-1)^{E_j}
 (u+1)^{G_j} }~du.  \label{eq:matrixelement} \end{equation}
Note that the exponents in the denominator of the integrand, $E_j$ and $G_j$, depend only on the column index, $j$.  Moreover, for the definition of $\mathcal{M}$ we require that each column have a similar structure regarding the exponents $r_{i,j}$, namely that
%for all $i,j'$:
\begin{equation}
    %r_{i+1,j} - r_{i,j} =
    %r_{i+1,j'} - r_{i,j'}.
    r_{i,j} = c_j + r_i
\end{equation}
where $c_j, r_i \in \mathbb Z$.  For the particular form of matrix we are interested in, given by (\ref{eq:matrixdef1}) to (\ref{eq:matrixdef3}), we could define, for example, $r_i=i-2$ for $1\leq i < K$, $r_K=K-1$ and $r_i=i-K-2$ for $K+1 \leq i\leq 2K$.  Then $c_j=j$ for $1\leq j\leq 2K$.

Let us define the degree of $M_{i,j}$ as
$$d_{i,j}=r_{i,j}-E_j-G_j .$$ We will also sometimes refer equivalently to the ``degree'' of $I(r_{i,j}, E_j, G_j)$.

We call the degree of the $J$th column 
$$D_J = D_J(M)=\max_i d_{i,J},$$
i.e.  the maximal degree of any entry in the column.
 We define the  total degree of $M$ to be   $$D=D(M)=\sum_{J=1}^{2K} D_J.$$
 
Note that any column with  $D_J\le -2 $ is a column of zeros. 

If we apply one of our recursion formulae, ~\eqref{eq:I recursion 1} or~\eqref{eq:I recursion 2}, to each entry in a particular  column
then  each entry in that column is a sum and we can split our determinant into
a sum of two determinants (see (\ref{detrule})) along that column, one determinant  will be of a
matrix with the same degree as the original matrix and one will have  a  degree
that is less by 1. This idea is utilised in the following lemma:

\begin{lemma} \label{lemma:twocolumnsequal}
If the matrix $M\in \mathcal M$ has two equal column degrees,  say
$D_J(M)=D_{J'}(M)$ for some $J\ne J'$,   then there exists a matrix $M_1 \in \mathcal M$ such that $\det M=\det M_1 $ and
$D(M_1)<D(M)$, i.e. we can replace the determinant in
question by a determinant of a matrix of lower degree.
\end{lemma}

\begin{proof}

We have $D_J(M)=D_{J'}(M)$ for some columns $J$ and $J'$.  Due to the structure of the exponent $r_{i,j}=r_i+c_j$ in (\ref{eq:matrixelement}), this means that $d_{i,J}=d_{i,J'}$ for all $1\leq i\leq 2K$.  Using the definition of the degree $d_{i,j}$, this means that $c_J+r_i-E_J-G_J=c_{J'}+r_i-E_{J'}-G_{J'}$ or 
\begin{equation}
c_{J'}=c_J-(E_J-E_{J'})-(G_J-G_{J'}).
\end{equation}

Assume for convenience that $E_J>E_{J'}$ and $G_J>G_{J'}$, but all other orderings follow in exactly the same way.  Then using ~\eqref{eq:I recursion 1} and~\eqref{eq:I recursion 2} we act on each element, indexed by $1\leq i\leq 2K$,  in column $J$ in the following way
\begin{eqnarray}
I(c_J+r_i, E_J,G_J)&=& I(c_J+r_i-1,E_J-1,G_J)+ lower = I(c_J+r_i-2,E_{J}-2,G_J)+lower =\cdots\nonumber \\
&=&I(c_J+r_i-(E_J-E_{J'}),E_{J'},G_J)+lower\\
&=& I(c_J+r_i-(E_{J}-E_{J'})-1, E_{J'},G_J-1)+lower \nonumber \\
&=& I(c_J+r_i-(E_J-E_{J'})-2,E_{J'},G_J-2)+lower =\cdots\nonumber \\
&=&I(c_J+r_i-(E_J-E_{J'})-(G_J-G_{J'}),E_{J'},G_{J'})+lower \nonumber\\
&=& I(c_{J'}+r_i,E_{J'},G_{J'})+lower,\nonumber
\end{eqnarray}
where $lower$ denotes a matrix element of lower degree.  This is true for any row $i$, so we separate the determinant, as described at (\ref{detrule}), so that we have the sum of two determinants, one with $I(c_J+r_i, E_J,G_J)$ replaced with $I(c_{J'}+r_i,E_{J'},G_{J'})$ in each element $(i,J)$ and the other with the $(i,J)$th element replaced by something of lower degree.  The former determinant is zero because it has two equal columns ($J$ and $J'$) and the latter is a determinant of a matrix of lower degree than $M$.
 \end{proof}
 
We continue, in the following two propositions, to eliminate cases where the determinant is zero. 
\begin{proposition}
\label{prop:deg det}
For $M\in \mathcal{M}$, suppose that $D(M) < 2K^2-3K$. Then $\det(M)=0$.  Furthermore, if $D(M)=2K^2-3K$ and $\det(M)\neq 0$ then it follows that the column degrees, in some order, take distinct values from $-1,0,1,2,\cdots,2K-2$.
\end{proposition}
 
\begin{proof}
We may assume that no two columns have equal degrees or else we apply Lemma \ref{lemma:twocolumnsequal} and
reduce out of that situation.  Next, if $D_J\leq-2$ for any $J$ then we have a column of zeros and the determinant is zero.
 Then the minimal total degree for a matrix with non-zero determinant  will occur
when the column degrees are (in some order) 
$-1,0,1,2, \dots , 2K-2$.
But
\begin{equation}
    -1+0+1+\dots +(2K-2) = 2K^2-3K.
\end{equation}
%Consequently, if the total degree is smaller then we must have two equal columns and can repeatedly apply the lemma until some column has degree smaller than $-1$ and so is 0. 
 \end{proof}
 Now we specialise to the case described by \eqref{eq:matrixdef1}-\eqref{eq:matrixdef3} with the following proposition.
 
\begin{proposition} \label{prop:nmsum} Suppose that    $m_j$ and $n_j$  are non-negative integers for $j=1, \dots ,2K$
such that 
\begin{equation}
    m_1+\dots +m_{2K}+n_1+\dots +n_{2K}>2K
\end{equation}
and let $M=(M_{i,j})_{1\le i,j\le 2K}$
with 
\begin{eqnarray} \label{eq:MijMatrix}
M_{i,j} = \left \{ \begin{array}{ll}
I(i+j-2, K+m_j, n_j) & \mbox{if $1 \leq i\le K-1, 1 \leq j\le 2K$}\\
I(i+j-1, K+m_j, n_j) & \mbox{if $i= K, 1\leq j\le 2K$}\\
I(i-K+j-2, m_j, K+n_j) & \mbox{if $ K+1\le i\le 2K,1\leq  j\le 2K$}\\
\end{array}
\right.
\end{eqnarray}
Then $\det M =0$.

The same is true if the matrix in question is
\begin{eqnarray} \label{eq:MijMatrix2}
M_{i,j} = \left \{ \begin{array}{ll}
I(i+j-2, K+m_j, n_j) & \mbox{if $1 \leq i\le K, 1 \leq j\le 2K$}\\
I(i-K+j-2, m_j, K+n_j) & \mbox{if $ K+1\le i\le 2K-1,1\leq  j\le 2K$}\\
I(i-K+j-1,m_j,K+n_j)&\mbox{if $i=2K,1\leq  j\le 2K$}
\end{array}
\right.
\end{eqnarray}
\end{proposition}
\begin{proof}
As the degree of $I(r_{i,j},E_j,G_j)$ is $r_{i,j}-E_j-G_j$, it is easy to check in (\ref{eq:MijMatrix}) that the maximal degree for each column comes from the 
entries in the $K$th row, and in (\ref{eq:MijMatrix2}) the maximal degree comes from entries in the $2K$th row. 
In either case, for the $J$th column, we have
\begin{equation}
    D_J(M)=J-1-m_J-n_J
\end{equation}
and 
\begin{equation}
    D(M)=\sum_{J=1}^{2K} D_J(M)=2K^2-K-\sum_{J=1}^{2K}(m_J+n_J)<2K^2-3K.
\label{<++>}
\end{equation}
By Proposition \ref{prop:deg det} we have $\det(M)=0$.
\end{proof}

Remembering that $m_j$ is the number of times that column $j$ has been differentiated with respect to $t_1$ and $n_j $ is the number of times column $j$ has been differentiated with respect to $t_2$, we have thus shown that all $(2K+1)$-st and higher partial derivatives of $\det M_3$ in Lemma~\ref{lem:2} are 0. Therefore $\det M_3$ is a polynomial of degree at most $2K$ in $t_1$
and $t_2$.  %**** PRESUMABLY THERE ARE SOME CONDITIONS NEEDED ON detM3 AS A FUNCTION OF $t_1$ and $t_2$ IN ORDER TO CONCLUDE THIS. DO WE NEED TO MENTION THEM OR IS IT OBVIOUS THAT THEY ARE SATISFIED?***

Next we determine that $\det M_3$ is a polynomial of degree $2K$ in $t_1$ and $t_2$ by identifying the coefficients of the terms $t_1^a t_2^b$ of degree $a+b=2K$. Consider
a mixed derivative $\tfrac{d^a}{d t_1^a} \;\tfrac{d^b}{dt_2^b}$ of $\det M_3$ and set $t_1=t_2=0$. As before, we get a sum of determinants,
where each determinant is associated to one of the ways in which we can differentiate the
columns of $\det M_3$ with respect to $t_1$ ($a$ times) and with respect to $t_2$ ($b$ times).
The following proposition describes what happens to a single one of these determinants.

\begin{proposition} \label{prop:coeff}
Now suppose that we have the same matrix $M$ defined at (\ref{eq:MijMatrix}) except with
\begin{equation}
m_1+\dots +m_{2K}+n_1+\dots +n_{2K}=2K,
\end{equation}
i.e. the total degree is $2K^2-3K$.
%Then $\det(M)=0$ if and only if either there are  two columns with the same degree or else there is  a column whose degree is $-2$ or smaller.  
If the determinant is not zero, then
\begin{equation}
    \det(M)=\pm \binom{2K-2}{K-1} 2^{(K-1)^2}.
\end{equation}
The same is true for a matrix of form (\ref{eq:MijMatrix2}).
\end{proposition} 
\begin{proof}
 
Let $m_J+n_J=p_J$ for each $1\leq J\leq 2K$.  Consider the top half of the matrix, $1 \leq i\leq K$.  
\begin{eqnarray}\label{eq:binomial1}
M_{i,J}&=&I(r_{i,J},K+m_J, n_J)=I(r_{i,J}-1,K+m_J,n_J-1)+lower \nonumber\\
&= &\cdots = I(r_{i,J}-n_J,K+m_J,0) +lower= I(r_{i,J}-n_J-1,K+m_J-1,0)+lower \\
&=& \cdots = I(r_{i,J}-n_J-m_J,K,0)+lower=I(r_{i,J}-p_J,K,0)+lower\nonumber\\
&=& \binom{r_{i,J}-p_J}{K-1} +lower, \nonumber
\end{eqnarray}
where the final evaluation is done by a simple residue calculation of the integral $I$.

In the lower half of the matrix, for $K+1\leq i\leq 2K$, we have similarly

\begin{eqnarray}\label{eq:binomial2}
M_{i,J}&=&I(r_{i,J},m_J, K+n_J)=I(r_{i,J}-1,m_J,K+n_J-1)+lower \nonumber\\
&= &\cdots = I(r_{i,J}-n_J,m_J,K) +lower= I(r_{i,J}-n_J-1,m_J-1,K)+lower \\
&=& \cdots = I(r_{i,J}-n_J-m_J,0,K)+lower=I(r_{i,J}-p_J,0,K)+lower\nonumber\\
&=& (-1)^{r_{i,J}-p_J-K-1}\binom{r_{i,J}-p_J}{K-1} +lower. \nonumber
\end{eqnarray}

Now we separate the determinant, as described at (\ref{detrule}), so that we have the sum of two determinants, one with the binomial coefficients down column $J$ and the other with a lower degree integral.  However, in the latter matrix, the degree of column $J$ will be lower than the degree of the original matrix $M$.  Since the degree of $M$ is $2K^2-3K$ and we ascertained in Proposition \ref{prop:deg det} that any matrix with lower degree has zero determinant, we are simply left with the determinant of the matrix with column $J$ replaced with the binomial coefficients given in  (\ref{eq:binomial1}) and (\ref{eq:binomial2}). We repeat this process for each of the columns of $M$ to end up with a matrix of binomial coefficients. 

We will now refer to the ``degree" of a binomial coefficient as being the degree of the integral $I$ that it came from.  So, the degree of $\binom{r_{i,J}-p_J}{K-1}$ is $r_{i,J}-p_J-K$.  Our matrix $M$ has the structure (\ref{eq:MijMatrix}) (or (\ref{eq:MijMatrix2})) and no degrees have been changed by the processes of turning it into a matrix of binomial coefficients using (\ref{eq:binomial1}) and (\ref{eq:binomial2}). Therefore the degree of a column is determined by the degree of the element in the $K$th row (respectively $2K$th).  In the $K$th (respectively $2K$th) row, $r_{K,J}=K+J-1$ ($r_{2K,J}=K+J-1$) so in either (\ref{eq:MijMatrix}) or (\ref{eq:MijMatrix2}) the degree of the $J$th column is $D_J=J-1-p_J$.  We know from Proposition \ref{prop:deg det} that since the degree of the matrix is still $2K^2-3K$, and the determinant is not zero, the column degrees for $J=1,2, \ldots, 2K$ must take distinct values in $-1,0,1,2,\cdots, 2K-2$.  For example, one way to arrange this would be to have $p_J=1$ for all $J$, but this is not the only solution.  Thus the elements of the $K$th ($2K$th) row {\em must}, in some order, take values $\binom{K-1}{K-1}$, $\binom{K}{K-1}, \cdots ,\binom{3K-2}{K-1}$ (resp. $\binom{K-1}{K-1}$, $-\binom{K}{K-1}, \binom{K+1}{K-1},\cdots ,-\binom{3K-2}{K-1}$ for matrix (\ref{eq:MijMatrix2})), so as to achieve the required set of column degrees.  If all the $p_J=1$ then the elements occur in this order across row $K$ ($2K$, respectively), but for other combinations of the $p_J$s they will occur in a different order. Once the set of $p_J$'s are fixed, then all the matrix entries are determined and we end up with a column-wise
permutation (implying an over all factor of $\pm 1$ that we haven't determined) of the matrix with entries that for an initial matrix (\ref{eq:MijMatrix}) look like
\begin{eqnarray}
\label{eq:m_ij}
m_{i,j} = \left \{ \begin{array}{ll}
\binom{i+j-3}{K-1}
  & \mbox{if $1 \leq i\le  K-1 , 1\leq j\le 2K$}\\
\binom{i+j-2}{K-1} &\mbox{if $i=  K ,1\leq  j\le 2K$}\\
(-1)^{i+j}\binom{i+j-3-K}{K-1}
 & \mbox{if $ K+1\le i\le 2K, 1\leq j\le 2K$}\\
\end{array}
\right.
\end{eqnarray}
or for an initial matrix (\ref{eq:MijMatrix2}) look like
\begin{eqnarray}
\label{eq:m_ij2}
m_{i,j} = \left \{ \begin{array}{ll}
\binom{i+j-3}{K-1}
  & \mbox{if $1 \leq i\le  K , 1\leq j\le 2K$}\\
(-1)^{i+j}\binom{i+j-3-K}{K-1}
 & \mbox{if $ K+1\le i\le 2K-1, 1\leq j\le 2K$}\\
(-1)^{j-1} \binom{i+j-2-K}{K-1} &\mbox{if $i=  2K ,1\leq  j\le 2K$}\\
\end{array}
\right..
\end{eqnarray}
Note that we take $\binom{-1}{K-1}=0$, so that the $(1,1)$ and $(K+1,1)$ entries of the matrix are $0$.

The $(K,1)$ entry of the matrix (\ref{eq:m_ij}) equals 1 while all the other entries in the first column are zero.
Expanding the determinant of the above matrix along the first column thus gives:
\begin{equation}
(-1)^{K+1}
\begin{vmatrix}
\binom{0}{K-1}&\cdots&\binom{2K-2}{K-1}\\
\vdots&\ddots&\vdots\\
\binom{K-2}{K-1}&\cdots&\binom{3K-4}{K-1}\\
(-1)^{K+1}\binom{0}{K-1}&\cdots&(-1)^{K+1}\binom{2K-2}{K-1}\\
\vdots&\ddots&\vdots\\
-\binom{K-2}{K-1}&\cdots&-\binom{3K-4}{K-1}\\
\binom{K-1}{K-1}&\cdots&\binom{3K-3}{K-1}
\end{vmatrix}.
\end{equation}
Next, we notice that the new first column is zero except the last entry. Expanding along that column we get
the following $(2K-2)\times(2K-2)$ determinant:
\begin{equation}\label{eq:matrixfirstcase}
(-1)^{K+1}
\begin{vmatrix}
\binom{1}{K-1}&\cdots&\binom{2K-2}{K-1}\\
\vdots&\ddots&\vdots\\
\binom{K-1}{K-1}&\cdots&\binom{3K-4}{K-1}\\
(-1)^{K}\binom{1}{K-1}&\cdots&(-1)^{K+1}\binom{2K-2}{K-1}\\
\vdots&\ddots&\vdots\\
\binom{K-1}{K-1}&\cdots&-\binom{3K-4}{K-1}\\
\end{vmatrix}.
\end{equation}
We arrive at the above matrix also for an initial matrix of form (\ref{eq:MijMatrix2}), but in that case the non-zero element of column 1 of (\ref{eq:m_ij2}) is a $+1$ in the $2K$th row, so expanding around that gives $-1$ times the resulting $(2K-1)\times(2K-1)$ minor.  In the  minor resulting from this first expansion, the non-zero element of the new first column is a $+1$ in the $K$th row, so expanding round this element give a sign of $(-1)^{K+1}$.  Thus in the (\ref{eq:MijMatrix2}) case we end up with the above determinant, but with the overall factor of $(-1)^{K+1}$ replaced by $(-1)^K$.

Working now from (\ref{eq:matrixfirstcase}) we apply row reductions using the identity (\ref{eq:pascal}) and exactly the same procedure as in equations (\ref{malphazero}) to (\ref{papermatrix}) so that we arrive at
\begin{equation}
(-1)^{K+1}\begin{vmatrix}
\binom{1}{K-1}&\cdots&\binom{2K-2}{K-1}\\
\binom{1}{K-2}&\cdots&\binom{2K-2}{K-2}\\
\vdots&\ddots&\vdots\\
\binom{1}{1}&\cdots&\binom{2K-2}{1}\\
(-1)^{K}\binom{1}{K-1}&\cdots&(-1)^{K+1}\binom{2K-2}{K-1}\\
(-1)^{K+1}\binom{1}{K-2}&\cdots&(-1)^K\binom{2K-2}{K-2}\\
\vdots&\ddots&\vdots\\
\binom{1}{1}&\cdots&-\binom{2K-2}{1}
\end{vmatrix}.
\end{equation}
Out of the $j$th column we factor out $j$, and we factor  $\frac{1}{(K-i)}$ out of both the $i$th row and the $(i+K-1)$th row  for $i=1,\dots K-1$. This gives the following
\begin{align}
&(-1)^{K+1}\binom{2K-2}{K-1}
\begin{vmatrix}
\binom{0}{K-2}&\cdots&\binom{2K-3}{K-2}\\
\binom{0}{K-3}&\cdots&\binom{2K-3}{K-3}\\
\vdots&\ddots&\vdots\\
\binom{0}{0}&\cdots&\binom{2K-3}{0}\\
(-1)^{K}\binom{0}{K-2}&\cdots&(-1)^{K+1}\binom{2K-3}{K-2}\\
(-1)^{K+1}\binom{0}{K-3}&\cdots&(-1)^K\binom{2K-3}{K-3}\\
\vdots&\ddots&\vdots\\
\binom{0}{0}&\cdots&-\binom{2K-3}{0}
\end{vmatrix}\\
&=\binom{2K-2}{K-1}2^{(K-1)^2}, \notag
\end{align}
where the last step follows by pulling out $(-1)^K$ from each of the bottom $K-1$ rows (hence an even power of $-1$), and
then applying~\eqref{papermatrix} with $K-1$ rather than $K$.

Hence the determinant of the matrix with entries given in~\eqref{eq:m_ij} is equal to
\begin{equation}
    \label{eq:det em}
    \binom{2K-2}{K-1} 2^{(K-1)^2}.
\end{equation}
The matrix with entries given in (\ref{eq:m_ij2}) is equal to 
\begin{equation}
    \label{eq:det em2}
    -\binom{2K-2}{K-1} 2^{(K-1)^2}.
\end{equation}
\end{proof}

For a given $d^a/d t_1^a d^b/d t_2^b$ (so $\sum_1^{2K} m_j=a$ and $\sum_1^{2K} n_j=b$, with $a+b=2K$), we now wish to determine 
what the multiplicity is of a given $(p_1,\ldots,p_{2K})$ where $p_j=m_j+n_j$. 

For example, if $K=2$, and $a=b=2$, the vector $(p_1,p_2,p_3,p_4)=(1,1,1,1)$ can arise in 24 ways. We have these patterns for $(m_1,m_2,m_3,m_4)$ $(n_1,n_2,n_3,n_4)$:
$(1,1,0,0)$ $(0,0,1,1)$;
$(1,0,1,0)$ $(0,1,0,1)$;
$(1,0,0,1)$ $(0,1,1,0)$;
$(0,1,1,0)$ $(1,0,0,1)$;
$(0,1,0,1)$ $(1,0,1,0)$;
$(0,0,1,1)$ $(1,1,0,0)$. However, each vector appearing here occurs twice
when we carry out the partial derivative $d^a/d t_1^a d^b/dt_2^b$ on the matrix $M_3$.  %, as we
%need to consider all the possible orderings in which we differentiate the columns of $M_3$ (so
For example, $(1,1,0,0)$ gets counted twice, as we can differentiate the first column and then the second,
or else the second column and then the first.  All the following arguments hold equally well if instead of matrix $M_3$ we use the matrix with the modified $2K$th row mentioned in Lemma \ref{lem:2}.

Generally, the number of occurrences of $(p_1,\ldots,p_{2K})$ obtained by applying $d^a/dt_1^a d^b/d t_2^b$
to $\det M_3$, is equal to
the coefficient of $c_1^{p_1} \cdots c_{2K}^{p_{2K}}$ in
\begin{equation}
    (c_1+\cdots +c_{2K})^a (c_1+\cdots +c_{2K})^b = (c_1+\cdots +c_{2K})^{2K}.
\end{equation}
The resulting coefficient therefore equals the multinomial coefficient
\begin{equation}
  \frac{  (2K)!}{\prod_{j=1}^{2K}p_j!}.
\end{equation}

Next we show that all of the $\pm1$ add up to 1.
The list of the degrees of the columns is a permutation of
\begin{equation}
    (-1,0,1,\dots,2K-2).
\end{equation}
If it is an {\it even} permutation then the sign will be plus; if it is an {\it
odd} permutation the sign will be minus.

As mentioned in the proof of Proposition \ref{prop:coeff}, any given permutation $\sigma$ of the sequence
$-1,\dots,2K-2$ completely determines the sequence of $p_j$. For example, when
$K=2$ there are 8 permutations which each give a determinant value of $\pm 4$.
The other 16 permutations give a determinant of 0. The lists of column degrees, $D_j$, the associated permutations, 
the corresponding sequence of $p_j$ that produce that permutation, along
with their signs and multiplicities are listed below. The sum of the
multiplicity times the sign gives $+1$ as desired.
\begin{eqnarray*}
\begin{array}{|rrrr|rrrr|rrrr|r|r|}
\hline 
&D_j&&&&\sigma& &&    &p_j&&&\mbox{sign}&\mbox{mult}\\
\hline
-1&0&1&2&  1&2&3&4&   1&1&1&1&+&24\\
-1&0&2&1&  1&2&4&3&   1&1&0&2&-&12\\
-1&1&0&2&  1&3&2&4&   1&0&2&1&-&12\\
-1&1&2&0&  1&3&4&2&   1&0&0&3&+&4\\
0&-1&1&2&  2&1&3&4&    0&2&1&1&-&12\\
0&-1&2&1&  2&1&4&3&    0&2&0&2&+&6\\
0&1&-1&2&  2&3&1&4&    0&0&3&1&+&4\\
0&1&2&-1&  2&3&4&1&    0&0&0&4&-1&1\\
\hline
\end{array}
\end{eqnarray*}

Inspecting the column of permutations, $\sigma$, we see that the legal permutations of $1,2,\dots, 2K$ are $\sigma_1,\sigma_2,$ $\dots,\sigma_{2K}$ with  $\sigma_1\le 2; \sigma_2\le 3; \sigma_3\le 4; \dots$.  The reason for this is that $p_j\geq 0$ and $p_j=j-1-D_j=j+1-\sigma_j$. The sign is just the sign of the permutation. The multiplicity is  $(2K)!/\prod_{j=1}^{2K}p_j!$.  So what we still have to prove is that
\begin{eqnarray}
\label{eq:final identity}
(2K)!\sum_{\sigma \in S_{2K}\atop \sigma_j\le j+1} \frac{\mbox{sgn}(\sigma)}{\prod_{j=1}^{2K}
(j+1-\sigma_j)!}=1
\end{eqnarray}

However, the sum in the above equation is the determinant of the $2K\times 2K$ matrix, denoted by $C_{2K}$, whose $(i,j)$ entry
is $1/(j+1-i)!$ if $i\leq j+1$, and 0 otherwise (because of the
restriction  $\sigma_j\leq j+1$). For example, the matrix $C_6$ equals:
\begin{equation}
    \left[ \begin {array}{cccccc} 1&1/2&1/6&1/24&1/120&1/720\\ \noalign{\medskip}1&1&1/2&1/
    6&1/24&1/120\\ \noalign{\medskip}0&1&1&1/2&1/6&1/24\\ \noalign{\medskip}0&0&1&1&1/2&1/6
    \\ \noalign{\medskip}0&0&0&1&1&1/2\\ \noalign{\medskip}0&0&0&0&1&1\end {array} \right] .
    \label{eq:K=3}
\end{equation}
%\left[ \begin {array}{cccc} 1&1/2&1/6&1/24\\ \noalign{\medskip}1&1&1/2&1/6\\ \noalign{\medskip}0&1&1&1/2
%\\ \noalign{\medskip}0&0&1&1\end {array} \right]

We can put $C_{2K}$ into triangular form by subtracting $i$ times row $i$ from row $i+1$, for $i=1,\ldots, 2K-1$.
Letting $l=j-i$, One can prove inductively, one row at a time, that the resulting entries are equal to 
$1/(l!(l+i))$ if $j\geq i$ and 0 otherwise. In particular the $(i,i)$ diagonal entry equals $1/i$, and hence
$\det C_{2K} = 1/(2K)!$, thus establishing the identity~\eqref{eq:final identity}.

We have thus proven that $d^a/d t_1^a d^b/d t_2^b$ applied to the matrix $M_3$,
and setting $t_1=t_2=0$, is equal to~\eqref{eq:det em}. Hence, the coefficient of $t_1^a t_2^b$ in
$\det M_3$ is equal to 
\begin{equation}
    \label{eq:ab coeff}
    \frac{1}{a! b!}
    \binom{2K-2}{K-1} 2^{(K-1)^2}.
\end{equation}
If instead of the matrix $M_3$, the matrix mentioned in Lemma \ref{lem:2} with the modified $2K$th column is used, then the coefficient of $t_1^a t_2^b$ in the determinant of that matrix is equal to, using (\ref{eq:det em2}), 
\begin{equation}
    \label{eq:ab coeff}
   - \frac{1}{a! b!}
    \binom{2K-2}{K-1} 2^{(K-1)^2}.
\end{equation}
Comparing coefficients, Lemma~\ref{lem:2} follows.

\section{Comparison with exact formula for $K=1$ and $K=2$} \label{sect:comparison}

In this section we work with Theorem \ref{theo:logderivexact}, in the two cases $|A|=|B|=1$ and $|A|=|B|=2$,  to show that this agrees with the result (\ref{eq:final asymptotic}) in
the appropriate limiting regime. 

First we consider $|A|=|B|=1$.  Writing out Theorem \ref{theo:logderivexact} in this case we have
\begin{eqnarray}
J(\{\alpha\};\{\beta\})&=&\int_{U(N)} (-e^{-\alpha}) \frac{\Lambda_X'}{\Lambda_X} (e^{-\alpha}) \; (-e^{-\beta}) \frac{\Lambda_{X^*}'}{  \Lambda_{X^*}} (e^{-\beta}) dX_N \nonumber \\
&=& H_{\{0\},\{0\}}(\{\alpha\})H_{\{0\},\{0\}}(\{\beta\}) +H_{\{0\},\{0\}}(\{\alpha\},\{\beta\})+e^{-N(\alpha+\beta)} z(\alpha+\beta)z(-\alpha-\beta) \nonumber \\
&=& 0+\left(\frac{z'}{z}\right)' (\alpha+\beta) +e^{-N(\alpha+\beta)} z(\alpha+\beta)z(-\alpha-\beta). 
\end{eqnarray}
Now let $\alpha=a/N$ and $\beta=b/N$ where $a,b\rightarrow 0$ as $N\rightarrow \infty$.  

It is useful for this and the $|A|=|B|=2$ calculation to write down the behaviour of $z(x)$ and its derivatives for small x:
\begin{eqnarray}\label{eq:zeds}
z(x)&=&\frac{1}{1-e^{-x}} =\frac{1}{x}+\frac{1}{2}+\frac{x}{12}-\frac{x^3}{720} +O(x^4)\nonumber \\
\frac{z'(x)}{z(x)}&=& \frac{1}{1-e^x} = -\frac{1}{x}+\frac{1}{2}-\frac{x}{12}+\frac{x^3}{720} +O(x^4) \nonumber\\
\left(\frac{z'(x)}{z(x)}\right)' &=& \frac{e^x}{(1-e^x)^2}=\frac{1}{x^2}-\frac{1}{12} +\frac{x^2}{240}+O(x^4).
\end{eqnarray}

Thus we have
\begin{eqnarray}
J(\{\tfrac{a}{N}\};\{\tfrac{b}{N}\}) &=& \left(\frac{1}{(\tfrac{a}{N}+\tfrac{b}{N})^2} +e^{-a-b} \left( \frac{-1}{(\tfrac{a}{N}+\tfrac{b}{N})^2} \right)\right)(1+O(\tfrac{a+b}{N})) \nonumber \\
&=&\left(\frac{N^2}{(a+b)^2} -(1-a-b)\frac{N^2}{(a+b)^2}\right)(1+O(a+b))\nonumber \\
&=&\left(\frac{N^2}{a+b}\right)(1+O(a+b)).
\end{eqnarray}

So
\begin{equation} \label{eq:J*1}
J^*(\{\tfrac{a}{N}\};\{\tfrac{a}{N}\}) = \frac{N^2}{2a}(1+O(a)),
\end{equation}
when $a=b$.  From the definition of $J^*$, and remembering that $\exp(\alpha)\sim 1$, we see that equation (\ref{eq:J*1}) is identical to (\ref{eq:final asymptotic}) when $K=1$. 

Now we consider 
\begin{eqnarray}
&&J(\{\alpha_1,\alpha_2\};\{\beta_1,\beta_2\})\nonumber \\
&&\qquad= \int_{U(N)} e^{-\alpha_1-\alpha_2-\beta_1-\beta_2} \frac{\Lambda_X'}{\Lambda_X} (e^{-\alpha_1})\frac{\Lambda_X'}{\Lambda_X} (e^{-\alpha_2})\frac{\Lambda_{X^*}'}{\Lambda_{X^*}} (e^{-\beta_1})\frac{\Lambda_{X^*}'}{\Lambda_{X^*}} (e^{-\beta_2 }) dX_N. 
\end{eqnarray}
We have to take a little care in setting all the alphas and betas equal here because we will encounter factors of $\frac{z'}{z}(\alpha_2-\alpha_1)$ and $\frac{z'}{z}(\beta_2-\beta_1)$.  These divergent terms will cancel as $\alpha_2\rightarrow \alpha_1$ and $\beta_2\rightarrow \beta_1$, but in order to control this we will set $\alpha_1=\beta_1=\alpha$ and $\alpha_2=\beta_2=\alpha+h$, with a view to letting $h\rightarrow 0$ later. 
\begin{eqnarray}
&&J(\{\alpha, \alpha+h\};\{\alpha,\alpha+h\})= \left(\frac{z'}{z}\right)'(2\alpha) \left(\frac{z'}{z}\right)'(2\alpha+2h) + \left(\frac{z'}{z}\right)'(2\alpha+h) \left(\frac{z'}{z}\right)'(2\alpha+h) \nonumber \\
&&\quad+ e^{-N(2\alpha)} z(2\alpha)z(-2\alpha) \left( \left(\frac{z'}{z}\right)'(2\alpha+2h) + \left( \frac{z'}{z} (h)-\frac{z'}{z}(2\alpha+h)\right) \left( \frac{z'}{z}(h)-\frac{z'}{z} (2\alpha+h)\right)\right) \nonumber \\
&&\quad+ e^{-N(2\alpha+h)} z(2\alpha+h)z(-2\alpha-h) \left( \left(\frac{z'}{z}\right)'(2\alpha+h) + \left( \frac{z'}{z} (h)-\frac{z'}{z}(2\alpha+2h)\right) \left( \frac{z'}{z}(-h)-\frac{z'}{z} (2\alpha)\right)\right) \nonumber \\
&&\quad+ e^{-N(2\alpha+h)} z(2\alpha+h)z(-2\alpha-h) \left( \left(\frac{z'}{z}\right)'(2\alpha+h) + \left( \frac{z'}{z} (-h)-\frac{z'}{z}(2\alpha)\right) \left( \frac{z'}{z}(h)-\frac{z'}{z} (2\alpha+2h)\right)\right) \nonumber \\
&&\quad+ e^{-N(2\alpha+2h)} z(2\alpha+2h)z(-2\alpha-2h) \left( \left(\frac{z'}{z}\right)'(2\alpha) + \left( \frac{z'}{z} (-h)-\frac{z'}{z}(2\alpha+h)\right) \left( \frac{z'}{z}(-h)-\frac{z'}{z} (2\alpha+h)\right)\right) \nonumber \\
&& \qquad \qquad+e^{-N(4\alpha+2h)} \frac{z(2\alpha) z^2(2\alpha+h)z(2\alpha+2h)z(-2\alpha)z^2(-2\alpha-h)z(-2\alpha-2h)}{ (z(-h)z(h))^2}.
\end{eqnarray}
The final term above is zero in the $h\rightarrow 0$ limit, but there are also terms of order $h^{-2}$ and order $h^{-1}$. Using Mathematica to expand to order $h^2$ anything multiplying the divergent terms, we can confirm that all divergent terms cancel.  In the $h\rightarrow 0$ limit we are left with
\begin{eqnarray}
J(\{\alpha,\alpha\};\{\alpha,\alpha\})&=& \lim_{h\rightarrow 0} J(\{\alpha,\alpha+h\};\{\alpha,\alpha+h\})\nonumber \\
&=&\frac{ 2e^{4\alpha} +e^{-2\alpha N}(-e^{2\alpha}N^2 + 2e^{4\alpha}N^2 -e^{6\alpha}N^2 -2e^{4\alpha})} {(1-e^{2\alpha})^4}.
\end{eqnarray}
Now we scale $\alpha=\tfrac{a}{N}$ where $a\rightarrow 0$ as $N\rightarrow \infty$.  Expanding the exponentials of the form $e^{ka/N}$, $k=2,4,6$, in powers of $a/N$, we find that terms in the numerator of order $N^2$ and $N$ cancel and we are left with:
\begin{eqnarray}
J(\{\tfrac{a}{N},\tfrac{a}{N}\};\{\tfrac{a}{N},\tfrac{a}{N}\})&=&\frac{\left(2+e^{-2a}(-4a^2+32a^2-36a^2-2)\right)N^4}{ 16a^4}\left(1+O(\tfrac{a}{N})\right)\nonumber \\
&=&\frac{(2+(1-2a)(-8a^2-2))N^4}{16a^4} \left(1+O(a)\right) = \frac{N^4}{4a^3} (1+O(a)).
\end{eqnarray}
Again, this is identical to (\ref{eq:final asymptotic}) when $K=2$.  We see that with the help of Mathematica, the leading order term of Theorem \ref{theo:J} can be extracted for a specific $K$.  However, obtaining a formula for a general $K$ seems very tricky from the complicated 
Theorem \ref{theo:J}, illustrating the value of the alternate method detailed in Section \ref{sect:theorem2} of this paper. 

%%%%%%%%%%%%%%%%%%%%%%%%%%%%%%%%%%%%%%%%%%%%%%%%%%%%%%%%%%%%%%

\section{Riemann-Hilbert problem representation for block Hankel determinant}
\label{sect:rhp}
In this section, we show that the determinant on the right-hand side of equation \eqref{eq:bhd} gives the solution of a certain Riemann-Hilbert problem. A block Hankel matrix has the form $[a_{j+k}]$ where all the submatrices $a_k$ have equal size. The $2K\times 2K$ moment determinant on the right-hand side of equation \eqref{eq:bhd} can be rearranged to become the determinant of a block Hankel matrix with $2\times 2$ blocks  by moving the $(K+j)^{\mathrm{th}}$ row to the $(2j)^{\mathrm{th}}$ place and then the $(K+j)^{\mathrm{th}}$ column to the $(2j)^{\mathrm{th}}$ place. Its matrix symbol is given by
\begin{equation}\arraycolsep=1.8pt\def\arraystretch{1.8}
w(u)=\left(\begin{array}{cc}
\frac{e^{\frac{au}{2}+\frac{t_1}{u-1}+\frac{t_2}{u+1}}}{(u-1)^K}&\frac{u^Ke^{\frac{au}{2}+\frac{t_1}{u-1}+\frac{t_2}{u+1}}}{(u-1)^K}\\
\frac{e^{-\frac{au}{2}+\frac{t_1}{u-1}+\frac{t_2}{u+1}}}{(u+1)^K}&\frac{u^Ke^{-\frac{au}{2}+\frac{t_1}{u-1}+\frac{t_2}{u+1}}}{(u+1)^K}
\end{array}\right).
\end{equation}

The underlying objects here are multiple orthogonal polynomials, and we refer the reader to \cite{kn:mfv,kn:k} for a review of them. For a review of integrable systems and Riemann-Hilbert problems, we refer the reader to \cite{kn:i1,kn:i2}. 
 The multiple orthogonal polynomials of type II are monic polynomials of degree $|\vec{n}|=n_1+n_2$ satisfying the conditions
\begin{equation}\label{eq:oc}
\intop_{|u|=2} P_{\vec{n}}(u) u^\ell w^{(j)}(u){du}=0,\quad 0\leq\ell\leq n_j-1,\quad 1\leq j\leq 2,\quad \vec{n}=(n_1,n_2)\in \mathbb{Z}^2,
\end{equation}
on the circle $|u|=2$, where the complex weights are given by
\begin{equation}
w^{(1)}(u)=\frac{e^{\frac{au}{2}+\frac{t_1}{u-1}+\frac{t_2}{u+1}}}{(u-1)^K},
\quad 
w^{(2)}(u)=\frac{e^{-\frac{au}{2}+\frac{t_1}{u-1}+\frac{t_2}{u+1}}}{(u+1)^K}.
\end{equation}

The polynomials $P_{\vec{n}}(u)$ admit  the determinant presentation
\begin{equation}\arraycolsep=1.4pt\def\arraystretch{1.7}\label{eq:detpres}
P_{\vec{n}}(u)=\dfrac{1}{\Delta_{\vec{n}}}\left|\begin{array}{ccc}
\mu_0^{(1)}&\cdots&\mu_{|\vec{n}|}^{(1)}\\
\vdots&&\vdots\\
\mu_{n_1-1}^{(1)}&\cdots&\mu_{|\vec{n}|+n_1-1}^{(1)}\\
\hline
\mu_0^{(2)}&\cdots&\mu_{|\vec{n}|}^{(2)}\\
\vdots&&\vdots\\
\mu_{n_2-1}^{(2)}&\cdots&\mu_{|\vec{n}|+n_2-1}^{(2)}\\
\hline
1&\hdots&u^{|\vec{n}|}
\end{array}
\right|,\quad \quad \Delta_{\vec{n}}=\left|\begin{array}{cccc}
\mu_0^{(1)}&\cdots&\mu_{|\vec{n}|-1}^{(1)}\\
\vdots&&\vdots\\
\mu_{n_1-1}^{(1)}&\cdots&\mu_{|\vec{n}|+n_1-2}^{(1)}\\
\hline
\mu_0^{(2)}&\cdots&\mu_{|\vec{n}|-1}^{(2)}\\
\vdots&&\vdots\\
\mu_{n_2-1}^{(2)}&\cdots&\mu_{|\vec{n}|+n_2-2}^{(2)}\\
\end{array}
\right|,
\end{equation}
where
\begin{equation}
\mu_{\ell}^{(j)}=\dfrac{1}{2\pi i}\intop_{|u|=2}  u^\ell {w^{(j)}(u)}{du}.
\end{equation}
The conditions \eqref{eq:oc} are easy to check using the linearity of the determinant with respect to the last row. Division by $\Delta_{\vec{n}}$ makes the polynomials monic. For generic $t_1$ and $t_2$, the determinant $\Delta_{\vec{n}}$ is nonzero. We define $\Delta_{(0,0)}=1$ and $P_{\vec{0}}(u)=1$. The determinant $\Delta_{(K,K)}$ is exactly the determinant appearing in \eqref{eq:bhd}.

The polynomials $P_{\vec{n}}(u)$ have a Riemann-Hilbert representation (see \cite{kn:gkv}). Actually, assume that
\begin{equation}\label{eq:detcond}
\Delta_{\vec{n}}\neq 0,\quad \Delta_{\vec{n}-\vec{e}_1}\neq 0,\quad \Delta_{\vec{n}-\vec{e}_2}\neq 0.
\end{equation}
Then we construct the matrix from the multiple orthogonal polynomials

\begin{equation}\arraycolsep=1.8pt\def\arraystretch{1.8}\label{eq:G}
\Gamma_{\vec{n}}(u)=\left(\begin{array}{ccc}P_{\vec{n}}(u)&R^{(1)}_{\vec{n}}(u)&R^{(2)}_{\vec{n}}(u)\\
b^{(1)}_{\vec{n}}P_{\vec{n}-\vec{e}_1}(u)&b^{(1)}_{\vec{n}}R^{(1)}_{\vec{n}-\vec{e}_1}(u)&b^{(1)}_{\vec{n}}R^{(2)}_{\vec{n}-\vec{e}_1}(u)\\
b^{(2)}_{\vec{n}}P_{\vec{n}-\vec{e}_2}(u)&b^{(2)}_{\vec{n}}R^{(1)}_{\vec{n}-\vec{e}_2}(u)&b^{(2)}_{\vec{n}}R^{(2)}_{\vec{n}-\vec{e}_2}(u)\\
\end{array}\right)
\end{equation}
where
\begin{equation}
R^{(j)}_{\vec{n}}(u)=\dfrac{1}{2\pi i}\intop_{|v|=2}  \dfrac{P_{\vec{n}}(v) {w^{(j)}(v)}}{v-u}{dv},\quad -\frac{1}{b^{(j)}_{\vec{n}}}=\dfrac{1}{2\pi i}\intop_{|u|=2} u^{n_j-1} {P_{\vec{n}-\vec{e}_j}(u) {w^{(j)}(u)}}{du}
\end{equation}
and $\vec{e}_1$ and $\vec{e}_2$ are the basis vectors in $\mathbb{Z}^2$. Conditions \eqref{eq:detcond} imply that $\frac{1}{b^{(j)}_{\vec{n}}}\neq 0$ and ${b^{(j)}_{\vec{n}}}\neq0$, because of the relations \eqref{eq:formula3}.

Orient the circle $|u|=2$ counter clockwise, to give the contour for the Riemann-Hilbert problem. As one goes round the contour, denote by ``$+$'' the boundary value from the left of the contour,  and by ``$-$''  the boundary value from the right of the contour. By the Sokhotski-Plemelj formula, we have
\begin{equation}
R^{(j)}_{\vec{n},+}(u)-R^{(j)}_{\vec{n},-}(u)=P_{\vec{n}}(u) {w^{(j)}(u)},\quad |u|=2.
\end{equation}
 This relation implies
\begin{equation}\label{eq:jump}
\Gamma_{\vec{n},+}(u)=\Gamma_{\vec{n},-}(u)\left(\begin{array}{ccc}1&{w^{(1)}(u)}&{w^{(2)}(u)}\\
0&1&0\\
0&0&1
\end{array}\right),\quad |u|=2.
\end{equation}
Because of the orthogonality conditions \eqref{eq:oc}, we have the asymptotic behavior
\begin{equation}
R^{(j)}_{\vec{n}-\vec{e_j}}(u)=-\dfrac{1}{u^{n_j}2\pi i}\intop_{|v|=2}  v^{n_j-1}{P_{\vec{n}}(v) {w^{(j)}(v)}}{dv}+O\left(\dfrac{1}{u^{n_j+1}}\right),\quad u\to\infty.
\end{equation}
Taking into account the normalizations by $b^{(1)}_{\vec{n}}$ and $b^{(2)}_{\vec{n}}$, this implies
\begin{equation}\label{eq:nai}
\Gamma_{\vec{n}}(u)=\left(I+\dfrac{\gamma^{(1)}_{\vec{n}}}{u}+O\left(\dfrac{1}{u^2}\right)\right)\left(\begin{array}{ccc}u^{|\vec{n}|}&0&0\\
0&u^{-n_1}&0\\
0&0&u^{-n_2}
\end{array}\right),\quad u\to \infty,
\end{equation}
where
\begin{equation} \label{eq:littlegamma}
\gamma^{(1)}_{\vec{n}}=\lim_{u\to\infty}u\cdot\left(\Gamma_{\vec{n}}(u)\cdot\mathrm{diag}[u^{-|\vec{n}|},u^{n_1},u^{n_2}] -I\right).
\end{equation}
We say that the function $\Gamma_{\vec{n}}(u)$, that is analytic inside and outside of the circle $|u|=2$ and that satisfies the conditions \eqref{eq:jump} and \eqref{eq:nai}, solves the Riemann-Hilbert problem. 
\begin{proposition} The Riemann-Hilbert problem with boundary conditions \eqref{eq:jump} and \eqref{eq:nai} has a unique solution, given by the above $\Gamma_{\vec{n}}(u)$ from equation \eqref{eq:G}.
\end{proposition}
\begin{proof}
We notice that $\det(\Gamma_{\vec{n}}(u))$  is an entire function, since it has no jump on the circle $|u|=2$. It converges to one at infinity; therefore $\det(\Gamma_{\vec{n}}(u))=1$ by Liouville's theorem. This means that we can always invert the matrix $\Gamma_{\vec{n}}(u)$.
Now we can show uniqueness of the solution of the Riemann-Hilbert problem. Assume there is some other solution $\widetilde{\Gamma}_{\vec{n}}(u)$. Then the function $\widetilde{\Gamma}_{\vec{n}}(u)\Gamma^{-1}_{\vec{n}}(u)$ has no jump on the circle $|u|=2$, is an entire function on the complex plane, and approaches the identity at infinity.  Therefore $\widetilde{\Gamma}_{\vec{n}}(u)={\Gamma}_{\vec{n}}(u)$ by Liouville's theorem.
\end{proof}

We proceed to state some recurrence relations \eqref{eq:formula1}, \eqref{eq:formula2}, and \eqref{eq:formula3} in terms of the Riemann-Hilbert problem. The recurrence relations for $\Delta_{\vec{n}}$ are based on ${\Gamma}_{\vec{n}}(u)$, and have the form
\begin{eqnarray}\label{eq:formula1}
\left(\gamma^{(1)}_{\vec{n}}\right)_{12}=-\dfrac{1}{2\pi i}\intop_{|u|=2} u^{n_1} {P_{\vec{n}}(u) {w^{(1)}(u)}}{du}=(-1)^{n_2+1}\dfrac{\Delta_{\vec{n}+\vec{e}_1}}{\Delta_{\vec{n}}},\\\label{eq:formula2}
\left(\gamma^{(1)}_{\vec{n}}\right)_{13}=-\dfrac{1}{2\pi i}\intop_{|u|=2} u^{n_2} {P_{\vec{n}}(u) {w^{(2)}(u)}}{du}=-\dfrac{\Delta_{\vec{n}+\vec{e}_2}}{\Delta_{\vec{n}}},\\\, \left(\gamma^{(1)}_{\vec{n}}\right)_{21} =b^{(1)}_{\vec{n}}=(-1)^{n_2+1}\dfrac{\Delta_{\vec{n}-\vec{e}_1}}{\Delta_{\vec{n}}},\quad\quad \label{eq:formula3}
 \left(\gamma^{(1)}_{\vec{n}}\right)_{31} =b^{(2)}_{\vec{n}}=-\dfrac{\Delta_{\vec{n}-\vec{e}_2}}{\Delta_{\vec{n}}}.
\end{eqnarray}
They can be checked using presentation \eqref{eq:detpres}. We can see that if $\left(\gamma^{(1)}_{\vec{n}}\right)_{12}=0$, then $\Delta_{\vec{n}+\vec{e}_1}=0$; and if $\left(\gamma^{(1)}_{\vec{n}}\right)_{13}=0$, then $\Delta_{\vec{n}+\vec{e}_2}=0$.

We also provide a more standard product formula \eqref{eq:formula4} for the determinant. In particular, since by definition $\Delta_{(0,0)}=1$, we have
\begin{equation}\label{eq:formula4}
\Delta_{(K,K)}=\prod_{j=1}^{K}(-1)^j\dfrac{\left(\gamma^{(1)}_{(j-1,j)}\right)_{12}}{\left(\gamma^{(1)}_{(j-1,j)}\right)_{31}}.
\end{equation}

We also identify the moment determinant \eqref{eq:bhd} with the isomonodromic tau function for a linear ODE  with rational coefficients by using the formula \eqref{eq:formula7} below; (see \cite{kn:b,kn:jmu,kn:jm}). 
To obtain this linear ODE  \eqref{eq:du}, we need to reformulate the Riemann-Hilbert problem so that it has constant a jump on a suitable contour. For this purpose, we introduce
\begin{equation}\label{eq:GG}
G(u)=\mathrm{diag}\left[\left(w^{(1)}(u)w^{(2)}(u)\right)^{\frac{1}{3}},\left(w^{(1)}(u)\right)^{-\frac{2}{3}}\left(w^{(2)}(u)\right)^\frac{1}{3},\left(w^{(2)}(u)\right)^{-\frac{2}{3}}\left(w^{(1)}(u)\right)^\frac{1}{3}\right].
\end{equation}
The function $G(u)$ has jumps along the rays $(-\infty,-1]$ and $[1,\infty)$ oriented towards infinity
$$
G_{+}(u)=G_{-}(u)e^{\frac{2\pi i K}{3}},\quad u\in(-\infty,-1]\cup[1,\infty);
$$
also $G(u)$ satisfies the relation
$$
G(u)^{-1}\left(\begin{array}{ccc}1&{w^{(1)}(u)}&{w^{(2)}(u)}\\
0&1&0\\
0&0&1
\end{array}\right)G(u)=\left(\begin{array}{ccc}1&1&1\\
0&1&0\\
0&0&1
\end{array}\right).
$$
Therefore the new function $\Psi_{\vec{n}}(u)=\Gamma_{\vec{n}}(u)G(u)$ has a constant jump on the circle $|u|=2$ oriented counter clockwise, and on the rays $(-\infty,-1]\cup[1,\infty)$ oriented towards infinity:
\begin{eqnarray}
&&\Psi_{\vec{n},+}(u)=\Psi_{\vec{n},-}(u)\left(\begin{array}{ccc}1&1&1\\
0&1&0\\
0&0&1
\end{array}\right),\quad |u|=2,
\\
&&
\Psi_{\vec{n},+}(u)=\Psi_{\vec{n},-}(u)e^{\frac{2\pi i K}{3}},\quad u\in(-\infty,-1]\cup[1,\infty).
\end{eqnarray}
The expression $\dfrac{d\Psi_{\vec{n}}(u)}{du}\Psi^{-1}_{\vec{n}}(u)$ has no jump on the circle $|u|=2$ and 
on the rays $(-\infty,-1]\cup[1,\infty)$. Therefore it is a rational function of $u$ with only poles at $u=1$,  $u=-1$, and $u=\infty$. Using \eqref{eq:nai} and \eqref{eq:GG}, we describe the local behavior  of function $\Psi_{\vec{n}}(u)$: at infinity, we have 
$$
\Psi_{\vec{n}}(u)=\left(I+\dfrac{Y_{\vec{n},1}^{(\infty)}}{u}+O\left(\dfrac{1}{u^2}\right)\right)e^{{d^{(\infty)}_{-1}u}}u^{d^{(\infty)}_{\vec{n},0}},\quad u\to\infty ,
$$

$$
d^{(\infty)}_{-1}=\mathrm{diag}\left[0,\frac{a}{2},-\frac{a}{2}\right],\quad 
d^{(\infty)}_{\vec{n},0}=\mathrm{diag}\left[|\vec{n}|-\frac{2K}{3},-n_1+\frac{K}{3},-n_2+\frac{K}{3}\right];
$$
near $u=1$, we have
$$
\Psi_{\vec{n}}(u)=P_{\vec{n}}^{(1)}\left(I+{Y_{\vec{n},1}^{(1)}}{(u-1)}+O\left({(u-1)^2}\right)\right)e^{\frac{d^{(1)}_{-1}}{u-1}}(u-1)^{d^{(1)}_0},\quad u\to1 ,
$$

$$
d^{(1)}_{-1}=\mathrm{diag}\left[\frac{2t_1}{3},-\frac{t_1}{3},-\frac{t_1}{3}\right],\quad 
d^{(1)}_0=\mathrm{diag}\left[-\frac{K}{3},-\frac{K}{3},\frac{2K}{3}\right];
$$
and near $n=-1$, we have
$$\arraycolsep=0pt\def\arraystretch{0.01}
\Psi_{\vec{n}}(u)=P_{\vec{n}}^{(-1)}\left(I+{Y_{\vec{n},1}^{(-1)}}{(u+1)}+O\left({(u+1)^2}\right)\right)e^{\frac{d^{(-1)}_{-1}}{u+1}}(u+1)^{d^{(-1)}_0},\quad u\to-1,
$$
$$
d^{(-1)}_{-1}=\mathrm{diag}\left[\frac{2t_2}{3},-\frac{t_2}{3},-\frac{t_2}{3}\right],\quad 
d^{(-1)}_0=\mathrm{diag}\left[-\frac{K}{3},\frac{2K}{3},-\frac{K}{3}\right].
$$
Here $P_{\vec{n}}^{(\pm 1)}$, $Y_{\vec{n},1}^{(\pm 1)}$,  $Y_{\vec{n},1}^{(\infty)}$ can be described by limits similar to \eqref{eq:littlegamma}. We can notice that $\det(\Psi_{\vec{n}}(u))=1$ implies $\det(P_{\vec{n}}^{(-1)})=\det(P_{\vec{n}}^{(1)})=1$, so we can compute inverses of these matrices. Using the local behavior of $\Psi_{\vec{n}}(u)$  we combine the principal parts of rational function $\dfrac{d\Psi_{\vec{n}}(u)}{du}\Psi^{-1}_{\vec{n}}(u)$  and we get the differential equation    \begin{equation}\label{eq:du}
\dfrac{d\Psi_{\vec{n}}(u)}{du}=A_{\vec{n}}(u)\Psi_{\vec{n}}(u),\quad A_{\vec{n}}(u)=A_{\vec{n}}^{(1)}(u)+A_{\vec{n}}^{(-1)}(u)+d^{(\infty)}_{-1},
\end{equation}
where the principal parts are 
\[
A_{\vec{n}}^{(1)}(u)=\dfrac{P_{\vec{n}}^{(1)}\Lambda^{(1)}_2(u)\left(P_{\vec{n}}^{(1)}\right)^{-1}}{(u-1)^2}+\dfrac{ P_{\vec{n}}^{(1)}\Lambda^{(1)}_{\vec{n},1}(u)\left(P_{\vec{n}}^{(1)}\right)^{-1}}{u-1},
\]
\[
\Lambda^{(1)}_2(u)=-{d^{(1)}_{-1}},\quad \Lambda^{(1)}_{\vec{n},1}(u)={d^{(1)}_0}-{[{Y_{\vec{n},1}^{(1)}},d^{(1)}_{-1}]},\]
\[
A_{\vec{n}}^{(-1)}(u)=\dfrac{P_{\vec{n}}^{(-1)}\Lambda^{(-1)}_2(u)\left(P_{\vec{n}}^{(-1)}\right)^{-1}}{(u+1)^2}+\dfrac{P_{\vec{n}}^{(-1)}\Lambda^{(-1)}_{\vec{n},1}(u)\left(P_{\vec{n}}^{(-1)}\right)^{-1}}{u+1},\]\[ \Lambda^{(-1)}_2(u)=-{d^{(-1)}_{-1}},\quad \Lambda^{(-1)}_{\vec{n},1}(u)={d^{(-1)}_0}-{[{Y_{\vec{n},1}^{(-1)}},d^{(-1)}_{-1}]}.\]
This is a differential equation with three irregular singular points of Poincar\'e rank one at $u=1$, $u=-1$ and $u=\infty$. The associated isomonodromic or Jimbo-Miwa-Ueno tau function (see \cite{kn:jmu}) is given by
\begin{equation}\label{eq:tau}
\tau_{\vec{n}}(a,t_1,t_2)=\exp\left(-\intop_{(0,0,0)}^{(a,t_{1},t_{2})} \mathrm{Tr}\left(Y_{\vec{n},1}^{(\infty)}d_{-1}^{(\infty)}\right)\dfrac{da}{a}+\mathrm{Tr}\left(Y_{\vec{n},1}^{(1)}d_{-1}^{(1)}\right)\dfrac{dt_1}{t_1}+\mathrm{Tr}\left(Y_{\vec{n},1}^{(-1)}d_{-1}^{(-1)}\right)\dfrac{dt_2}{t_2}\right).
\end{equation}
It was proven in \cite{kn:jmu} that the differential form that is integrated in \eqref{eq:tau} is closed. In the later work \cite{kn:jm}, the following relations were obtained
$$
\delta\left(\ln\dfrac{\tau_{\vec{n}+\vec{e_1}}}{\tau_{\vec{n}}}\right)=\delta\left(\ln\left(Y_{\vec{n},1}^{(\infty)}\right)_{12}\right),\quad \delta\left(\ln\dfrac{\tau_{\vec{n}+\vec{e_2}}}{\tau_{\vec{n}}}\right)=\delta\left(\ln\left(Y_{\vec{n},1}^{(\infty)}\right)_{13}\right),
$$
where $\delta$ is the differential with respect to $a$, $t_1$ and $t_2$.
Taking into account the identities  
$$
\left(Y_{\vec{n},1}^{(\infty)}\right)_{12}=\left(\gamma^{(1)}_{\vec{n}}\right)_{12}, \quad \left(Y_{\vec{n},1}^{(\infty)}\right)_{13}=\left(\gamma^{(1)}_{\vec{n}}\right)_{13},
$$
and using formulae \eqref{eq:formula1}, \eqref{eq:formula2} we obtain
$$
\delta\left(\ln\dfrac{\tau_{\vec{n}+\vec{e_1}}\Delta_{\vec{n}}}{\tau_{\vec{n}}\Delta_{\vec{n}+\vec{e_1}}}\right)=0,\quad \delta\left(\ln\dfrac{\tau_{\vec{n}+\vec{e_2}}\Delta_{\vec{n}}}{\tau_{\vec{n}}\Delta_{\vec{n}+\vec{e_2}}}\right)=0.
$$
Therefore 
\begin{equation}\label{eq:formula5}
\ln\dfrac{\Delta_{\vec{n}}(a,t_1,t_2)}{\tau_{\vec{n}}(a,t_1,t_2)}=\ln{h^{(1)}_{\vec{n}}}-\ln h^{(2)}_{(a,t_1,t_2)},
\end{equation}
where $h^{(1)}_{\vec{n}}$ is independent of $a$, $t_1$, $t_2$ and $h^{(2)}_{(a,t_1,t_2)}$ is independent of $\vec{n}$.

 If we put $\vec{n}=0$ in \eqref{eq:formula5}, then we have 
\begin{equation}\label{eq:tauzerozero}
\dfrac{ h^{(2)}_{(a,t_1,t_2)}}{h^{(1)}_{\vec{0}}}=\tau_{\vec{0}}(a,t_1,t_2)
\end{equation}
It turns out that we can compute $\tau_{\vec{0}}(a,t_1,t_2)$. Actually, we have $P_{\vec{0}}(u)=1$ and
\begin{equation}\label{eq:GGG}
\Gamma_{\vec{0}}(u)=\left(\begin{array}{ccc}1&R^{(1)}_{\vec{0}}(u)&R^{(2)}_{\vec{0}}(u)\\
0&1&0\\
0&0&1\\
\end{array}\right).
\end{equation}
The diagonal part of first terms in the expansion of $\Psi_{\vec{0}}(u)$ near singular points $u=\infty$, $u=1$ and $u=-1$ are
$$
\left(Y_{\vec{0},1}^{(\infty)}\right)_{\mathrm{diag}}=\mathrm{diag}\left[\frac{2}{3}(t_1+t_2),-\frac{1}{3}(t_1+t_2)-K,-\frac{1}{3}(t_1+t_2)+K\right],
$$
$$
\left(Y_{\vec{0},1}^{(1)}\right)_{\mathrm{diag}}=\mathrm{diag}\left[-\frac{1}{6}(t_2+K),-\frac{a}{2}+\frac{t_2}{12}-\frac{K}{6},\frac{a}{2}+\frac{t_2}{12}+\frac{K}{3}\right],
$$
$$
\left(Y_{\vec{0},1}^{(-1)}\right)_{\mathrm{diag}}=\mathrm{diag}\left[-\frac{1}{6}(t_1-K),-\frac{a}{2}+\frac{t_1}{12}-\frac{K}{3},\frac{a}{2}+\frac{t_2}{12}+\frac{K}{6}\right].
$$
Plugging this into \eqref{eq:tau}, we get
\begin{equation}\label{eq:tauzero}
\tau_{\vec{0}}(a,t_1,t_2)=\exp\left(-Ka+\dfrac{K}{6}(t_2-t_1)-\dfrac{t_1t_2}{6}\right).
\end{equation}

 If we put $a=t_1=t_2=0$ in \eqref{eq:formula5}, then we obtain
\begin{equation}\label{eq:deltazero}
\dfrac{h^{(1)}_{\vec{n}}}{h^{(1)}_{\vec{0}}}=\Delta_{\vec{n}}(0,0,0).
\end{equation}
In the case $\vec{n}={(K,K)}$, such a determinant was computed in Lemma \ref{lem:1}  \begin{equation}\label{eq:ck}\Delta_{(K,K)}(0,0,0)=(-2)^{K^2}.\end{equation} 
Combining \eqref{eq:tauzerozero}, \eqref{eq:tauzero}, \eqref{eq:deltazero}, and \eqref{eq:ck} together, we deduce the relation  between the block Hankel determinant and the isomonodromic tau function  (see \cite{kn:b})
\begin{equation}\label{eq:formula7}\Delta_{(K,K)}=(-2)^{K^2}\exp\left(Ka-\dfrac{K}{6}(t_2-t_1)+\dfrac{t_1t_2}{6}\right){\tau_{(K,K)}}.\end{equation}

%\bibliographystyle{../../../bin/ninplain}
%\bibliography{../../../bin/ref}

\end{document}